\documentclass[11pt,a4paper]{article}
\usepackage[T1]{fontenc}
\usepackage[utf8]{inputenc}
\usepackage[english]{babel}
\usepackage{lmodern,amsmath,amssymb,amsthm,mathtools,bm,braket}
\usepackage[margin=27mm]{geometry}
\usepackage{microtype,booktabs,enumitem}
\usepackage[colorlinks=true,linkcolor=blue,citecolor=blue,urlcolor=blue]{hyperref}
\usepackage{array}
\usepackage{comment}
\usepackage{authblk}

\usepackage{orcidlink}

\newcommand{\M}{\mathsf M_d}
\newcommand{\id}{\operatorname{id}}
\newcommand{\im}{\operatorname{Im}}
\newcommand{\tr}{\operatorname{Tr}}
\newcommand{\norm}[1]{\left\lVert#1\right\rVert}
\newcommand{\Ad}{\operatorname{Ad}}
\newtheorem{theorem}{Theorem}[section]
\newtheorem{proposition}[theorem]{Proposition}
\newtheorem{lemma}[theorem]{Lemma}
\newtheorem{corollary}[theorem]{Corollary}
\newtheorem{remark}[theorem]{Remark}
\newtheorem{example}[theorem]{Example}

\usepackage{ulem, xcolor}
\newcommand{\fs}[1]{\textcolor{blue}{#1}}
\newcommand{\dc}[1]{\textcolor{red}{#1}}

\begin{document}

\title{Divisibility of Non-Invertible Quantum Dynamics\\
in the Heisenberg Picture}

\author{Federico Settimo\, \orcidlink{0000-0002-0123-6950}$^1$ and Dariusz Chru{\'s}ci{\'n}ski\, \orcidlink{0000-0002-6582-6730}$^2$}
%\footnote{Corresponding author: darch@fizyka.umk.pl} 
\affil{$^1$ Department of Physics and Astronomy, University of Turku, \\ FI-20014 Turun Yliopisto, Finland \\
$^2$Institute of Physics, Faculty of Physics, Astronomy and Informatics, \\Nicolaus Copernicus University, Grudzi\c{a}dzka 5/7, 87--100 Toru{\'n}, Poland\\
}
\date{}
\maketitle

\begin{abstract}

We study divisibility of finite-dimensional quantum dynamics in the
Heisenberg picture when the dynamical maps need not be invertible.  Heisenberg divisible 
dynamics determines unique reduced propagators on the accessible operator
systems, and Heisenberg CP-divisibility is shown to be equivalent to an 
operator-norm monotonicity.
Although each reduced propagator admits a unital completely positive extension, independently chosen extensions need not  satisfy a composition law.
Generalized inverses restore
this law algebraically, but do not, in general, preserve complete positivity.
We characterize the resulting obstruction and give a sufficient construction
based on unital completely positive projections onto the accessible operator
systems. It is shown that for qubit dynamics Heisenberg
CP-divisibility always guarantees a completely positive
generalized-inverse propagator. Finally, we
analyze time-local generators at singular times and identify the associated
gauge freedom on dynamically inaccessible directions.
 
\end{abstract}

\section{Introduction}

No quantum system is perfectly isolated from its surroundings. Interaction with an external environment, a measuring apparatus, or uncontrolled degrees of freedom generally transforms the unitary evolution of the total system into a nonunitary evolution of the subsystem of interest \cite{Breuer-Petruccione, Rivas-Huelga-OQS, Vacchini-OQS}. In the Schrödinger picture this reduced dynamics is described by a family of completely positive and trace-preserving (CPTP) maps $\{\Lambda_t\}_{t\geq0}$ acting on states. Equivalently, in the Heisenberg picture one considers the dual family $ \Lambda^*_t$, consisting of unital completely positive (UCP) maps acting on observables. The two pictures give the same expectation values,

$$
\operatorname{Tr}\bigl[\Lambda_t(\rho)X\bigr]
 =
\operatorname{Tr}\bigl[\rho\,\Lambda^*_t(X)\bigr],
$$
but  the intermediate dynamics is different in the two pictures.

An important problem in the theory of open quantum systems is to determine whether the evolution between two times can itself be represented by a physical map. This leads to the notion of divisibility. The standard Schr\"odinger divisibility requires the existence of intermediate maps $V_{t,s}$ such that $\Lambda_t = V_{t,s} \Lambda_s$ for any $t > s$. Similarly, Heisenberg divisibility requires $\Lambda^*_t = W_{t,s}^* \Lambda^*_s$ \cite{Settimo-SchroHeis}. If for all $t>0$ each map $\Lambda_t$ (equivalently $\Lambda^*_t$) is invertible, then both Schr\"odinger and Heisenberg divisibility are guaranteed; indeed $V_{t,s} = \Lambda_t \Lambda_s^{-1}$ and $W_{t,s}^* = \Lambda^*_t {\Lambda^*_s}^{-1}$. We call the evolution Schr\"odinger CP-divisible whenever each map $V_{t,s}$ is 
CPTP, and Heiseberg CP-divisible iff $W_{t,s}^*$ is UCP. Schr\"odinger CP-divisibility has become one of the standard notions of quantum Markovianity: it expresses the idea that the evolution can be continued from any intermediate time without requiring information that has already been lost to the environment \cite{rivas_quantum_2014, DeVega2017, BLPV, li_concepts_2018, Chruscinski2022}. Interestingly, in a recent paper \cite{Settimo-SchroHeis}, it was shown that in general Schr\"odinger and Heisenberg CP-divisibility are independent properties of a quantum evolution. 

In this paper we analyze CP-divisibility for a general quantum evolution which is represented by dynamical maps which might be non-invertible. Non-invertible dynamics maps were already used to represent evolution of quantum-optical models with structured reservoirs or single-photon wave packets \cite{andersson_finding_2007, cresser_master_2010, hou_singularity_2012, hegde_open_2021, fang_non-markovian_2018, dabrowska_eternally_2021}. Mathematical properties of particular models were analyzed for example in \cite{siudzinska_non-markovianity_2022, jagadish_measure_2022, jagadish_noninvertibility_2023} (cf. also recent works \cite{davidovic_nonperturbative_2026, cockett_generalized_2026, blumenfeld_modifying_2025}). Authors of \cite{CRS} proved that for a $d$-level {system} if 

\begin{equation}   \label{CRS-d}
  \bigl\|(\operatorname{id}_d\otimes\Lambda_t)(X)\bigr\|_1
\leq
\bigl\|(\operatorname{id}_d\otimes\Lambda_s)(X)\bigr\|_1,
\end{equation}
for all $X=X^* \in \M \otimes \M$, then there exist a family  of {\it{reduced}} CPTP propagators $v_{t,s} : \im \Lambda_s \to \im \Lambda_t$. They can be extended to the full matrix algebra $V_{t,s} :\M \to \M$ as CP maps but in general trace-preservation is not guaranteed. Interestingly, it turned out \cite{chakraborty_information_2019} that for qubits ($d=2$) Schr\"odinger CP-divisibility is equivalent to (\ref{CRS-d}). In what follows we analyze a similar problem but for Heisenberg CP-divisibility. Our central result characterizes this property directly in terms of the loss of distinguishability of observables. Heisenberg CP-divisibility is equivalent to

\begin{equation}   \label{!}
   \bigl\|(\operatorname{id}_d\otimes\Lambda^*_t)(X)\bigr\|_\infty
\leq
\bigl\|(\operatorname{id}_d\otimes\Lambda^*_s)(X)\bigr\|_\infty
\end{equation}
for all $X=X^* \in \M \otimes \M$. Condition (\ref{!}) implies the existence of a family  of {\it reduced} UCP propagators $w_{t,s}^* : \im \Lambda^*_s \to \im \Lambda^*_t$. Now, contrary to the Schr\"odinger case, due to the Arveson extension theorem \cite{Arveson}, they can be extended to the full matrix algebra $W_{t,s}^* :\M \to \M$ as UCP maps.

The reduced propagators $w_{t,s}^*$ are uniquely defined by the dynamical map $\Lambda^*_t$  
%intrinsic \fs{\it (not sure I understand what you mean by intrinsic)} 
and satisfy the expected composition law $w_{t,u}^* w_{u,s}^* = w_{t,s}^*$ on the accessible operator system $\im \Lambda^*_s$. Their extensions $W_{t,s}^*$ to the full algebra, however, are highly nonunique and independently chosen extensions need not compose. This leads to a second problem studied in the paper: how to construct a consistent family of full propagators once invertibility is lost. Generalized inverses provide an explicit algebraic answer. Given a generalized inverse ${\Lambda^*_s}^{-}$, the maps

$$
{W_{t,s}^G}^*=\Lambda^*_t{\Lambda^*_s}^{-}
$$
are unital linear propagators and automatically satisfy the composition rule. Complete positivity, however, is not automatic. Such propagators are constrained to take values in
$\operatorname{Im}\Lambda^*_t$, whereas a general completely positive extension of the reduced propagator may have to enlarge this range. We show that a useful sufficient condition is the existence of a UCP projection onto each accessible operator system. Conditional expectations, Tomiyama's theorem, and the Choi--Effros product clarify the meaning of this complementability condition. This additional requirement is genuinely stronger than CP-divisibility. We construct a continuous Heisenberg CP-divisible evolution on $\mathsf M_3$, starting from the identity, for which no choice of generalized inverse produces a completely positive propagator for certain pairs of times. Hence replacing an inverse by a generalized inverse does not, by itself, solve the completely positive propagation problem. A similar problem was analyzed for Schr\"odinegr CP-divisibility in \cite{CC}. 
The situation is more rigid for a qubit. We prove that{, for} a norm-continuous Heisenberg CP-divisible qubit evolution starting from the identity, UCP generalized-inverse propagators can always be constructed.  The qubit is therefore special both for Schr\"odinger and Heisenberg CP-divisibility property.

We supplement the general results with explicit examples. They demonstrate that Schrödinger and Heisenberg divisibility may differ even at the level of existence of linear propagators, and that a UCP extension may be forced to leave the range of the final dynamical map. For singular phase-covariant qubit dynamics we derive exact finite-time and differential divisibility conditions in both pictures. We also discuss {existence and explicit construction of }time-local generators.

The paper is organized as follows.  Section~\ref{sec:schrodinger-heisenberg} reviews Schr\"odinger and
Heisenberg divisibility for invertible dynamics.  Section~\ref{sec:divisibility_non_inv} treats the
noninvertible case: it establishes the kernel and image criteria, proves the
operator-norm characterization of Heisenberg CP-divisibility, analyzes
generalized-inverse propagators, and isolates the special qubit mechanism.
Section~\ref{sec:tl_gen} discusses time-local generators and their nonuniqueness at singular
times.  Section~\ref{sec:conclusions} summarizes the conclusions and open questions.  The appendix
contains the continuous qutrit counterexample to complete positivity of the
generalized-inverse construction.

\section{Schr\"odinger versus Heisenberg divisibility}
\label{sec:schrodinger-heisenberg}

The Schr\"odinger and Heisenberg pictures give equivalent physical
predictions, but they do not attach the same divisibility property to a
generic time-dependent dynamical map.  The reason is elementary but
important: taking the Hilbert--Schmidt adjoint reverses the order of
composition.  Following Ref.~\cite{Settimo-SchroHeis}, we first make this point in
the invertible setting, where all intermediate maps are unique, and then
identify what changes when invertibility is lost.

Let $\Lambda_t:\M\to\M$ be a family of  {CPTP} maps, with $\Lambda_0=\id$, and let
$\Lambda^*_t$ denote the dual Heisenberg evolution.  Thus
$\Lambda^*_t$ is {UCP}, and
\begin{equation}
 \tr[\Lambda_t(\rho)X]=\tr[\rho\,\Lambda^*_t(X)] .
% \qquad (\rho,X\in\M).
 \label{eq:Sch-Heis-duality}
\end{equation}
Throughout the paper invertibility means
invertibility as a linear map on $\M$.  Neither $\Lambda_t^{-1}$ nor
${\Lambda^*_t}^{-1}$ is assumed to be positive.

If $\{\Lambda_t\}_{t \geq 0}$ is invertible for all $t \geq 0$, there are two natural two-parameter maps,
\begin{equation}
 V_{t,s}:=\Lambda_t \Lambda_s^{-1},
 \qquad
 W_{t,s}:=\Lambda_s^{-1} \Lambda_t,
 \qquad t\geq s,
 \label{eq:left-right-propagators}
\end{equation}
and hence two factorizations,
\begin{equation}
 \Lambda_t
 =V_{t,s} \Lambda_s
 =\Lambda_s  W_{t,s}.
 \label{eq:left-right-factorizations}
\end{equation}
Both propagators -- left $V_{t,s}$ and right $W_{t,s}$ -- are trace preserving and Hermiticity preserving, but
they need not be positive.  The usual Schr\"odinger definition uses the
left factorization: the dynamics is Schr\"odinger CP-divisible when
$V_{t,s}$ is CPTP for every $t\geq s$.  This choice has the
direct state-evolution interpretation $\rho_t =V_{t,s}(\rho_s$.

Since $(A  B)^*=B^*  A^*$, one has in the Heisenberg picture
\begin{equation}
    \Lambda^*_t
    =W^*_{t,s} \Lambda^*_s
    =\Lambda^*_s  V^*_{t,s},
    \label{eq:Heis-left-right-factorizations}
\end{equation}
so that the adjoint of the left Schr\"odinger propagator is a right propagator in the Heisenberg picture, and viceversa.
Since it holds that $X_t = W^*_{t,s}(X_s)$ for every evolved observable $X_s=\Lambda^*_s(X)$, Heisenberg CP-divisibility means that the left propagator $W^*_{t,s}$ is UCP for all $t\geq s$.
%\sout{The corresponding Heisenberg propagators are
%\begin{equation}
 %W^*_{t,s}:=\Lambda^*_t {\Lambda^*_s}^{-1},
 %\qquad
 %V^*_{t,s}:={\Lambda^*_s}^{-1} \Lambda^*_t,
 %\label{eq:Heis-left-right-propagators}
%\end{equation}
%so that
%\begin{equation}
 %\Lambda^*_t
 %=W^*_{t,s} \Lambda^*_s
 %=\Lambda^*_s  V^*_{t,s}.
 %\label{eq:Heis-left-right-factorizations}
%\end{equation}
%Heisenberg CP-divisibility means that the left propagator
%$W^*_{t,s}$ is UCP for all $t\geq s$; equivalently,
%$X(t)=W^*_{t,s}(X(s))$ for every evolved observable
%$X(s)=\Lambda^*_s(X)$.  Since $(A  B)^*=B^*  A^*$, one has
%\begin{equation}
 %W^*_{t,s}=(w_{t,s}^*)^*,
 %\qquad
 %V^*_{t,s}=(V_{t,s})^*.
 %\label{eq:duality-reverses-propagators}
%\end{equation}
Consequently,
\begin{align}
 &\text{Schr\"odinger CP-divisibility}
 &&\Longleftrightarrow \ \
 V_{t,s}\text{ is CPTP}
 &&\Longleftrightarrow \ \
 V^*_{t,s}\text{ is UCP},
 \label{eq:Sch-left-Heis-right}\\
 &\text{Heisenberg CP-divisibility}
 &&\Longleftrightarrow \ \ 
 W^*_{t,s}\text{ is UCP}
 &&\Longleftrightarrow \ \ 
 W_{t,s}^*\text{ is CPTP}.
 \label{eq:Heis-left-Sch-right}
\end{align}
Thus duality identifies a left factorization in one picture with a right
factorization in the other.  It does {not} identify the two left
divisibility conditions.  Although Eq.~\eqref{eq:Sch-Heis-duality}
describe the same expectation values, Schr\"odinger and Heisenberg
divisibility test, respectively, the complete positivity of
$\Lambda_t\Lambda_s^{-1}$ and of $\Lambda_s^{-1}\Lambda_t$.  These maps need not to coincide.
The two families obey natural composition laws as
\begin{equation}
    V_{t,u}  V_{u,s}
 =V_{t,s},
 \qquad
 W^*_{t,u}  W^*_{u,s}
 =W^*_{t,s}.
 \label{eq:left-right-composition}
\end{equation}
 For an invertible family $\{\Lambda_t\}_{t\ge 0}$, Schr\"odinger P-divisibility is
equivalent to \cite{chruscinski_measures_2011,CRS}

\begin{equation}
 \norm{\Lambda_t(X)}_1\leq\norm{\Lambda_s(X)}_1
 \qquad(t\geq s,\;X=X^*),
 \label{eq:Sch-trace-norm-monotonicity}
\end{equation}
whereas Heisenberg P-divisibility is equivalent to
\begin{equation}
 \norm{\Lambda^*_t(Y)}_\infty\leq\norm{\Lambda^*_s(Y)}_\infty
 \qquad(t\geq s,\;Y=Y^*).
 \label{eq:Heis-operator-norm-monotonicity}
\end{equation}
The trace norm quantifies the distinguishability of states (more
generally, of Hermitian Helstrom matrices), while the operator norm
quantifies the distinguishability of effects.  Thus a revival of the
trace-norm quantity witnesses failure of Schr\"odinger P-divisibility,
whereas a revival of the operator-norm quantity witnesses failure of
Heisenberg P-divisibility~\cite{Settimo-SchroHeis}.

Complete positivity is obtained by imposing the corresponding
inequality at matrix level {on extensions of the Hilbert space}.  In dimension $d$, it is enough to use a
$d$-dimensional ancilla:
\begin{align}
 \norm{(\id_d\otimes\Lambda_t)(X)}_1
 &\leq
 \norm{(\id_d\otimes\Lambda_s)(X)}_1,
 & X&=X^*,
 \label{eq:Sch-complete-trace-norm}\\
 \norm{(\id_d\otimes\Lambda^*_t)(Y)}_\infty
 &\leq
 \norm{(\id_d\otimes\Lambda^*_s)(Y)}_\infty,
 & Y&=Y^*.
 \label{eq:Heis-complete-operator-norm}
\end{align}
For invertible dynamics, the first condition characterizes
Schr\"odinger CP-divisibility and the second characterizes Heisenberg
CP-divisibility.  These finite-time inequalities are  robust
formulations.  Differential versions require additional regularity:
under local absolute continuity the relevant norm functions are
absolutely continuous and their derivatives are nonpositive almost
everywhere.

%whereas under mere continuity a pointwise formulation must
%use an upper Dini derivative.

%%%%%%%%%%%%%%%%%%%%%%%%%%%%%%%%%%%%%%%%%%%%%%%%%%%%%%%%%%%%%%%%%
\section{Divisibility of non-invertible dynamics}
\label{sec:divisibility_non_inv}

Suppose now that the family $\{\Lambda_t\}_{t \ge 0}$ allows for non-invertible maps for some times. 
At such {\it singular times}, the formulas containing $\Lambda_s^{-1}$ or
${\Lambda^*_s}^{-1}$ cease to be available, and an intermediate map is no
longer unique. Schr\"odinger divisibility and Heisenberg divisibility mean, respectively,
\begin{align}
 \Lambda_t =V_{t,s}\Lambda_s, \quad  \Lambda^*_t = W_{t,s}^*\Lambda^*_s, 
\end{align}
where $V_{t,s}, W_{t,s}^* : \M\to\M$. Moreover, evolution is  Schr\"odinger CP-divisible  (Heisenberg CP-divisible) if $V_{t,s}$ is CPTP ($W_{t,s}$ is UCP)  for $t \geq s$. P-divisibility requires positivity instead of complete positivity. Note, that if the map is not invertible at {time $s$} the existence of propagators $V_{t,s}$ and $W_{t,s}^*$ is not guaranteed.

\begin{proposition}[{Schr\"odinger divisibility, }\cite{CRS}]
A family $\{\Lambda_t\}_{t\ge 0}$ of CPTP maps is Schr\"odinger divisible if and only if
\begin{equation}\label{eq:kernel-S}
 \ker\Lambda_s\subseteq\ker\Lambda_t , \qquad(t\geq s).
\end{equation}
\end{proposition}
In this case the reduced propagator
\begin{equation}\label{eq:reduced_Schro}
 v_{t,s}: \im\Lambda_s \longrightarrow\mathcal  \im\Lambda_t ,
\end{equation}
defined by

\[  v_{t,s}(\Lambda_s(X))=\Lambda_t(X) ,  \]
are uniquely determined, Hermiticity preserving, and trace-preserving (on $\im\Lambda_s$). These reduced propagators $v_{t,s}$ can be extended to $V_{t,s} : \M \to \M$ which are Hermiticity preserving \cite{CRS}

\[  v_{t,s} = V_{t,s}\Big|_{\im\Lambda_s} . \]
Trace-preservation, instead, is guaranteed to hold only on $\im\Lambda_s$ and might not hold on the whole $\M$ \cite{heinosaari_extending_2012}. 
Of course the extension $V_{t,s}$ is highly non unique.

\begin{proposition}[Heisenberg divisibility]\label{prop:kernel}
A family $\Lambda^*_t$ of UCP maps is  Heisenberg divisible if and only if
\begin{equation}\label{eq:kernel}
 \ker\Lambda^*_s\subseteq\ker\Lambda^*_t\qquad(t\geq s).
\end{equation}
In this case the reduced propagator
\begin{equation}\label{eq:reduced}
 w_{t,s}^*:\mathcal S_s\longrightarrow\mathcal S_t,\qquad
 w_{t,s}^*(\Lambda^*_s(X))=\Lambda^*_t(X),\qquad \mathcal S_s:=\im\Lambda^*_s,
\end{equation}
is uniquely determined, Hermiticity preserving, and unital. It has a unital Hermiticity-preserving, unital linear extension {$W_{t,s}^*$} to $\M$.
\end{proposition}
\begin{proof}

Factorization immediately implies~\eqref{eq:kernel}. Conversely, assume \eqref{eq:reduced}.
For $Y\in\mathcal S_s$, choose $X\in\M$ such that $Y=\Lambda^*_s(X)$, and define 

\[ w_{t,s}^*(Y):=\Lambda^*_t(X). \] 
This definition does not depend on the choice of $X$. Indeed, if $\Lambda^*_s(X)=\Lambda^*_s(X')$, 
then 

\[ X-X'\in\ker\Lambda^*_s\subseteq\ker\Lambda^*_t, \] 
and consequently $\Lambda^*_t(X)=\Lambda^*_t(X')$.  Thus \(w_{t,s}^*\) is well defined {and unique}. Moreover, its range is exactly \(\mathcal S_t\). 
Since \(\Lambda^*_s\) and \(\Lambda^*_t\) are unital, 

\[ w_{t,s}^*(I) =w_{t,s}^*\bigl(\Lambda^*_s(I)\bigr) =\Lambda^*_t(I) =I. \] 
Hence \(w_{t,s}^*\) is unital. To  prove that \(w_{t,s}^*\) is Hermiticity preserving. Let \(Y=Y^*\in\mathcal S_s\), and choose \(X\in\M\) such that \(\Lambda^*_s(X)=Y\). Since \(\Lambda^*_s\) is Hermiticity preserving, 

\[ \Lambda^*_s\left(\frac{X+X^*}{2}\right) = \frac{\Lambda^*_s(X)+\Lambda^*_s(X)^*}{2} =Y. \] 
Therefore the Hermitian operator $H:=\frac 12 (X+X^*)$  
is also a preimage of \(Y\). It follows that $w_{t,s}^*(Y)=\Lambda^*_t(H)$, which is Hermitian because \(\Lambda^*_t\) is Hermiticity preserving. Thus \(w_{t,s}^*\) is Hermiticity preserving. 
Finally, to extend \(w_{t,s}^*\) from \(\mathcal S_s\) to all of \(\M\) define $W^*_{t,s} := w_{t,s}^* P_s$, where $P_s$ is an arbitrary Hemiticity-preserving linear projection onto $\mathcal S_s$. 
\end{proof}

Using the well known relations \cite{Horn}

\begin{equation}\label{eq:orthogonal}
 \im \Lambda^*=(\ker \Lambda)^\perp,\qquad \im \Lambda=(\ker \Lambda^*)^\perp ,
\end{equation}
one observes that already at the level of linear maps Schr\"odinger and Heisenberg divisibility have different existence conditions:

\begin{align}
 \Lambda_t={V_{t,s}} \Lambda_s
 &\quad\Longleftrightarrow\quad
 \ker\Lambda_s\subseteq\ker\Lambda_t
 \quad\Longleftrightarrow\quad
 \im\Lambda^*_t\subseteq\im\Lambda^*_s,
 \label{eq:singular-Sch-condition}\\
 \Lambda^*_t={W_{t,s}^*} \Lambda^*_s
 &\quad\Longleftrightarrow\quad
 \ker\Lambda^*_s\subseteq\ker\Lambda^*_t
 \quad\Longleftrightarrow\quad
 \im\Lambda_t\subseteq\im\Lambda_s.
 \label{eq:singular-Heis-condition}
\end{align}
Whenever the corresponding kernel inclusion holds, the dynamics fixes
a unique reduced propagator on the accessible range,
\begin{align}
 v_{t,s}:\im\Lambda_s&\longrightarrow\im\Lambda_t,
 &v_{t,s}(\Lambda_s(X))&:=\Lambda_t(X),
 \label{eq:reduced-Sch-propagator}\\
 w_{t,s}^*:\im\Lambda^*_s&\longrightarrow\im\Lambda^*_t,
 &w_{t,s}^*(\Lambda^*_s(X))&:=\Lambda^*_t(X).
 \label{eq:reduced-Heis-propagator}
\end{align}
These maps satisfy their composition rules intrinsically on the
accessible ranges: for $t \ge u \ge s$

\begin{equation}
    \label{eq:comp_law_reduced_props}
    v_{t,s} = v_{t,u}   v_{u,s} \ , \quad  w_{t,s}^* = w_{t,u}^*   w_{u,s}^* .
\end{equation}
Indeed, for every $X\in\M$, one finds for Heisenberg propagators (the same applies for $v_{t,s}$)
\[
 w_{t,u}^*  w_{u,s}^*(\Lambda^*_s(X))
 =w_{t,u}^*(\Lambda^*_u(X))
 =\Lambda^*_t(X)
 =w_{t,s}^*(\Lambda^*_s(X)).
\]
This identity is an equality of maps on $\mathcal S_s$.  If
$W_{t,s}^*:\M\to\M$ is an arbitrary extension of the reduced maps, then
$W_{t,u}^*  W_{u,s}^*$ and $W_{t,s}^*$ necessarily agree on $\mathcal S_s$,
but they need not agree on all of $\M$.  Thus pairwise UCP
extendibility does not by itself make independently chosen extensions
a cocycle.  The distinction between the intrinsic maps $w_{t,s}^*$ and
their nonunique extensions $W_{t,s}^*$ will be used throughout.

%Full propagators on $\M$ --- 
%of $V_{t,s}$ and $w_{t,s}^*$ ---  are generally nonunique; independently
%chosen extensions need not satisfy a composition law on $\M$. Note, however, that

%\begin{equation}
%    V_{t,s}\Big|_{\im \Lambda_s} = V_{t,u}   V_{u,s}\Big|_{\im \Lambda_s}   \quad {\rm and} \quad  %w_{t,s}^*\Big|_{\im \Lambda^*_s} = w_{t,u}^*   w_{u,s}^*\Big|_{\im \Lambda^*_s} ,
%\end{equation}
%that is, the local composition law is restored on the image of $\Lambda_s$ (or $\Lambda^*_s$). 

%%%%%%%%%%%%%%%%%%%%%%%%%%%%%%%%%%%%%%%%%%%

\subsection{Heisenberg CP-divisibility}

%\fs{\it Subsection?}

%\section{A complete operator-norm characterization}
Each $\mathcal S_s$ is a complex linear, self-adjoint subspace containing $I$, and thus an operator system with the matrix order inherited from $\M$. Let us recall two well known properties \cite{Paulsen,Bhatia,Stormer}
(for completeness we provide simple proofs).

\begin{proposition}\label{lem:order}

Let $\mathcal S$ be an operator system and $\varphi:\mathcal S\to\mathsf M_d$ a unital Hermiticity-preserving linear map. Then $\varphi$ is positive if and only if
\begin{equation}
 \norm{\varphi(Y)}_\infty\leq\norm{Y}_\infty \qquad (Y=Y^*\in\mathcal S).
\end{equation}
\end{proposition}

\begin{proof}
Positivity and unitality preserve $-rI\leq Y\leq rI$. Conversely, for $0\leq Y\in\mathcal S$, $Y\ne0$, set $Z=Y/\norm{Y}_\infty$. Then $\norm{I-Z}_\infty\leq1$, so $\norm{I-\varphi(Z)}_\infty\leq1$. Since $\varphi(Z)$ is Hermitian, all its eigenvalues belong to $[0,2]$. Thus $\varphi(Y)\geq0$. \end{proof}

Applying the same argument to $\id_k\otimes \varphi$ we obtain characterisation of $k$-positive unital maps.

\begin{proposition} \label{lem:dpositive}
A $d$-positive map from an operator system $\mathcal{S}$ into $\mathsf M_d$ is completely positive.

\end{proposition}

\begin{proof}
Let $Y\in\mathsf M_n(\mathcal S)$ be positive. Any vector $\xi\in\mathbb C^n\otimes\mathbb C^d$ has Schmidt rank $r\leq d$, so $\xi=(V\otimes I)\eta$ for an isometry $V:\mathbb C^r\to\mathbb C^n$. The compression $(V^*\otimes I)Y(V\otimes I)$ lies in $\mathsf M_r(\mathcal S)_+$. Since $d$-positivity implies $r$-positivity, the expectation of $(\id_n\otimes \varphi)(Y)$ in $\xi$ is nonnegative. This holds for every $\xi$ and every $n$.
\end{proof}

\begin{theorem}[Heisenberg CP-divisibility]\label{thm:main}

Let $\Lambda^*_t:\M\to\M$ be UCP. The following are equivalent:
\begin{enumerate}[label=(\roman*),nosep]
\item For every $t\geq s$ there exists a UCP map $W_{t,s}^*$ with $\Lambda^*_t=W_{t,s}^*\Lambda^*_s$, 

\item Kernels satisfy~\eqref{eq:kernel}, and every reduced map $w_{t,s}^*$ is completely positive.

\item For every $t\geq s$ and every Hermitian $X\in\mathsf M_d\otimes\M$,

\begin{equation}\label{eq:complete-norm}
 \norm{(\id_d\otimes\Lambda^*_t)(X)}_\infty
 \leq\norm{(\id_d\otimes\Lambda^*_s)(X)}_\infty.
\end{equation}
\end{enumerate}
Equivalently, the norm inequality holds with an ancilla of arbitrary finite dimension.

\end{theorem}
\begin{proof}
A UCP propagator is contractive on Hermitian matrices at every {amplification} ${\rm id}_k \otimes W^*_{t,s}$, proving (i)$\Rightarrow$(iii). Assume (iii). Testing {Eq.~\eqref{eq:complete-norm} on operators of the form} $|1\rangle\langle1|\otimes X$ gives unamplified contractivity. A Hermitian element annihilated by $\Lambda^*_s$ is therefore annihilated by $\Lambda^*_t$. Splitting an arbitrary kernel element into Hermitian and anti-Hermitian parts proves~\eqref{eq:kernel}.

Define $w_{t,s}^*$ by~\eqref{eq:reduced}. Every Hermitian $Y\in\mathsf M_d(\mathcal S_s)$ has a Hermitian preimage $X$ under $\id_d\otimes\Lambda^*_s$: lift its entries and then take the Hermitian part of the resulting block matrix. Consequently,
\[
 \norm{(\id_d\otimes w_{t,s}^*)(Y)}_\infty
 =\norm{(\id_d\otimes\Lambda^*_t)(X)}_\infty
 \leq\norm{Y}_\infty.
\]
Proposition~\ref{lem:order} proves $d$-positivity, and Lemma~\ref{lem:dpositive} proves complete positivity. This establishes {(iii)$\Rightarrow$(ii).} 
{Lastly, assume (ii).} Due to Arveson's theorem \cite{Arveson} any UCP map $\varphi : \mathcal{S} \to \M$ can be extended to a UCP map $\Phi : \M \to \M$. Therefore, we may extend each $w_{t,s}^*:\mathcal S_s\to\M$ to a UCP map $W_{t,s}^*:\M\to\M$. Since $I\in\mathcal S_s$, the extension satisfies $W_{t,s}^*(I)=w_{t,s}^*(I)=I$. {Therefore, (i) follows.}
\end{proof}

This proof isolates the difference from the Schr\"odinger construction of Ref.~\cite{CRS}: the normalization condition $W_{t,s}^*(I)=I$ is prescribed by the value of a single element $I \in \mathcal{S}_s$, already in the domain of extension. It makes no assertion that taking adjoints of a Schr\"odinger reduced propagator solves the same factorization problem. Indeed, it was shown \cite{CRS} that each CPTP reduced propagator $v_{t,s} : \im \Lambda_s \to \M$ can be extended to a CP map $V_{t,s} : \M \to \M$; however,  trace-preservation is no longer guaranteed (of course $V_{t,s}$ is trace-preserving on $\im \Lambda_s$). 

In \cite{CRS} it was shown that if $\{\Lambda_t\}_{t \ge 0}$ is image non-increasing, that is $\im \Lambda_t \subseteq \im \Lambda_s$ (for $t > s$), then it is Schr\"odinger CP-divisible if and only if 

\begin{equation} \label{S-CPTP}
\norm{(\id_d\otimes\Lambda_t)(X)}_1
 \leq\norm{(\id_d\otimes\Lambda_s)(X)}_1 ,
\end{equation}
for all Hermitian $X \in \M \otimes \M$. Now, the condition $\im \Lambda_t \subseteq \im \Lambda_s$ is equivalent to $\ker\Lambda^*_s \subseteq \ker\Lambda^*_t$, and hence we arrive at the following

\begin{corollary} If the family $\{\Lambda_t\}_{t \ge 0}$ satisfies (\ref{S-CPTP}) and it is Heisenberg divisible (i.e. $\ker\Lambda^*_s \subseteq \ker\Lambda^*_t$),   then  it is Schr\"odinger CP-divisible.
\label{cor:inclusion_Kers}
\end{corollary}

The finite-time inequality~\eqref{eq:complete-norm} is the 
statement.  Assuming that $\Lambda^*_t$ is locally absolutely continuous, Heisenberg CP-divisibility is equivalent to

\begin{equation}   
    \frac{d}{dt}\, \norm{(\id_d\otimes\Lambda^*_t)(X)}_\infty \leq 0 , 
\end{equation}
for almost all $t >0$. 

\begin{remark}
The monotonicity condition

\begin{equation} \label{H-PU}
\norm{\Lambda^*_t(X)}_\infty
 \leq\norm{\Lambda^*_s(X)}_\infty ,
\end{equation}
is equivalent to kernel inclusion \eqref{eq:kernel} together with positivity of all reduced maps $w_{t,s}^*$, by Proposition~\ref{prop:kernel} and Proposition~\ref{lem:order}. Heisenberg P-divisibility implies this condition. Its converse additionally requires positive unital extensions. Arveson's theorem does not supply such extensions for a merely positive reduced map. Positive-map extension problems require separate hypotheses~\cite{Stormer}, {see \cite{Chiribella} for} an example of unextendible norm one positive map on an operator system. 

\end{remark}

%\end{document}

{We now provide two examples of dynamics that are CP-divisible in one picture, no linear propagator exists in the dual picture because of violations of the kernel inclusion condition of Eq.~\eqref{eq:kernel}.}

\begin{example}[Schr\"odinger CP-divisible but not Heisenberg divisible]

Consider a qubit dynamical map

\begin{equation}
 \Lambda_t(\rho)=
 \begin{cases}
 (1-t)\rho+t\tr(\rho)I/2,&0\leq t\leq1,\\
 \tr(\rho)\,{\omega_t},&t\geq1,
 \end{cases}
 \qquad
 {\omega_t}=\tfrac12[I+r(t)\sigma_z],
\end{equation}
where $r$ is continuous, $r(1)=0$, and $0\leq r(t)\leq1$, with $r(t)>0$ at some later time. Before $t=1$, CPTP propagators are depolarizing channels with parameter $(1-t)/(1-s)$. For every $t\geq1$ and $s<t$, the reset channel $X\mapsto\tr(X){\omega_t}$ is a CPTP propagator. The full family is therefore Schr\"odinger CP-divisible.

In contrast, $\Lambda^*_1(\sigma_z)=0$, whereas $\Lambda^*_t(\sigma_z)=r(t)I$ {for $t>1$}. Kernel inclusion fails. The family is not even linearly Heisenberg divisible. Physically, once all initial-state information is erased, Schr\"odinger divisibility allows the prepared state to change; Heisenberg picture left divisibility forbids a vanished observable from reappearing.
\end{example}

%\end{document}

\begin{example}[Heisenberg CP-divisible but not Schr\"odinger divisible]
    
Let {$\mathcal E_z(\rho) \coloneqq (\rho+\sigma_z\rho\sigma_z)/2$ denote complete dephasing map} in the $\sigma_z$ basis. Define the Heisenberg maps by
\begin{equation}
 \Lambda^*_t=
 \begin{cases}
 (1-t)\id+t\mathcal E_z,&0\leq t\leq1,\\
 \Ad_{U_t}  \mathcal E_z,&t\geq1,
 \end{cases}
 \qquad \Ad_U(X)=UXU^*,
\end{equation}
where $U_t$ is continuous and $U_1=I$. For $s<t<1$, 
{%\sout{a dephasing channel with transverse multiplier $(1-t)/(1-s)$ is a UCP propagator} 
the propagator $W^*_{t,s}$ is UCP}. Across $t=1$, use { $W^*_{t,1}=\Ad_{U_t}\mathcal E_z$ }, since $\mathcal E_z\Lambda^*_s=\mathcal E_z$. For $t\geq s\geq1$, use $W_{t,s} = \Ad_{U_tU_s^*}$. Hence the whole family is Heisenberg CP-divisible.
Its Heisenberg range after $t=1$ is
\begin{equation}
 \mathcal S_t=\operatorname{span}\{I,U_t\sigma_zU_t^*\}.
\end{equation}
If $U_t$ rotates the $z$ axis, these two-dimensional ranges are not nested $\im\Lambda_t^*\not\subseteq\im\Lambda^*_s$.  Eq.~\eqref{eq:singular-Sch-condition} %\sout{Corollary~\ref{cor:both}
then rules out linear Schr\"odinger divisibility.  Thus, the two pictures are not equivalent regarding the existence of linear propagators.
\end{example}

\subsection{Generalized-inverse construction of propagators}

Let ${\Lambda_s^*}^-$ be a generalized inverse of $\Lambda^*_s$ \cite{GI_1,GI_2,GI_3}
\begin{equation}\label{eq:inner}
 \Lambda^*_s {\Lambda_s^*}^- \Lambda^*_s=\Lambda^*_s.
\end{equation}
No positivity or unitality of ${\Lambda_s^*}^-$ is assumed. The map $P_s=\Lambda^*_s{\Lambda_s^*}^-$ is a projection onto $\mathcal S_s$. Indeed, one has

\[
P_s^2
 =\Lambda^*_s{\Lambda^*_s}^{-}\Lambda^*_s{\Lambda^*_s}^{-}
 =\Lambda^*_s{\Lambda^*_s}^{-}
 =P_s .
\]

\begin{proposition}\label{prop:ginv}
Under kernel inclusion condition (\ref{eq:kernel}), for any choice of ${\Lambda_s^*}^-$ at each time,
\begin{equation}\label{eq:ginv}
 {W_{t,s}^G}^*:=\Lambda^*_t{\Lambda_s^*}^-
\end{equation}
is a unital linear propagator satisfying
\begin{equation}\label{eq:cocycleG}
 {W^G_{t,u}}^*{W^G_{u,s}}^*={W_{t,s}^G}^*,\qquad t\geq u\geq s.
\end{equation}
Moreover,
\begin{equation}\label{eq:ginv-range}
 {W_{t,s}^G}^*=w_{t,s}^*P_s,\qquad
 \im {W_{t,s}^G}^*=\mathcal S_t,\qquad W^G_{s,s}=P_s.
\end{equation}
\end{proposition}
\begin{proof}
Equation~\eqref{eq:inner} implies $(\id-{\Lambda_s^*}^-\Lambda^*_s)(\M)\subseteq\ker\Lambda^*_s$. Kernel inclusion yields $\Lambda^*_t{\Lambda_s^*}^-\Lambda^*_s=\Lambda^*_t$. Evaluating this identity at $I$ proves unitality

\[ {W_{t,s}^G}^*(\Lambda^* _s(I)) = \Lambda^*_t(I) \ \Longrightarrow \  {W_{t,s}^G}^*(I)=I . \]
Moreover, $\Lambda^*_t{\Lambda_s^*}^-\Lambda^*_s=\Lambda^*_t$ implies
\[
 \Lambda^*_t {\Lambda^*_u}^- \Lambda^*_u{\Lambda_s^*}^-=\Lambda^*_t{\Lambda_s^*}^-,
\]
which proves the composition law \eqref{eq:cocycleG}. Since $w_{t,s}^*\Lambda^*_s=\Lambda^*_t$, one has 

$$w_{t,s}^*P_s=\Lambda^*_t{\Lambda_s^*}^- = {W_{t,s}^G}^* . $$
The range of ${W_{t,s}^G}^*$ is contained in $\mathcal S_t$ and contains it because ${W_{t,s}^G}^*\Lambda^*_s=\Lambda^*_t$.
\end{proof}
These algebraic identities parallel the generalized-inverse construction of Ref.~\cite{CC}. Unitality here follows without imposing a normalization condition on the complementary subspace. Hermiticity preservation follows if the generalized inverse is chosen to be Hermiticity preserving. Obviously, complete positivity does not follow from~\eqref{eq:ginv}.

\begin{proposition}\label{prop:all}
For a fixed generalized inverse ${\Lambda_s^*}^-$, every {linear} solution of $W\Lambda^*_s=\Lambda^*_t${, with $\Lambda^*_t$ obeying the kernel inclusion \eqref{eq:kernel},} has the form

\begin{equation}\label{eq:all}
 W=\Lambda^*_t{\Lambda_s^*}^-+\Psi (\id-P_s),\qquad P_s=\Lambda^*_s{\Lambda_s^*}^-,
\end{equation}
with $\Psi : \M \to \M$ an arbitrary linear map. Every such solution is unital $W(I)=I$.
\end{proposition}
\begin{proof}
Since \(P_s\) is a projection onto \(\mathcal S_s=\operatorname{Im}\Lambda^*_s\) one has 

\[
(\operatorname{id}-P_s)\Lambda^*_s=0.
\]
Now let \(\Psi:M_d\to M_d\) be arbitrary linear map and define
\[
W_\Psi
 :=\Lambda^*_t{\Lambda^*_s}^{-}
   +\Psi(\operatorname{id}-P_s).
\]
Then
\begin{align*}
W_\Psi\Lambda^*_s
 =\Lambda^*_t{\Lambda^*_s}^{-}\Lambda^*_s
   +\Psi(\operatorname{id}-P_s)\Lambda^*_s  =\Lambda^*_t,
\end{align*}
so every map of the form (\ref{eq:all}) is a solution.

Conversely, suppose that \(W\Lambda^*_s=\Lambda^*_t\). Then
\[
WP_s
 =W\Lambda^*_s{\Lambda^*_s}^{-}
 =\Lambda^*_t{\Lambda^*_s}^{-},
\]
and therefore
\begin{align*}
W
 =WP_s+W(\operatorname{id}-P_s) =\Lambda^*_t{\Lambda^*_s}^{-}
   +W(\operatorname{id}-P_s).
\end{align*}
Thus \(W\) has the form \eqref{eq:all}, with \(\Psi=W\). Finally, every such solution is unital. Indeed, since both
\(\Lambda^*_s\) and \(\Lambda^*_t\) are unital,
\[
W_\Psi(I)
 =W_\Psi\bigl(\Lambda^*_s(I)\bigr)
 =\Lambda^*_t(I)
 =I.
\]
Equivalently, \(I=\Lambda^*_s(I)\in\mathcal S_s\), so
\(P_s(I)=I\) and the arbitrary complementary term
\(\Psi(\operatorname{id}-P_s)\) vanishes at \(I\).
\end{proof}

Note, that the parametrization (\ref{eq:all}) is generally redundant: only the restriction of \(\Psi\) to \(\ker P_s\) affects \(W_\Psi\). {Conversely}, one may ask which extensions $w_{t,s}^*$ admit {a representation via a generalized inverse ${\Lambda^*_s}^-$.}

\begin{proposition}\label{prop:representation} 
Let $\ker\Lambda^*_s\subseteq\ker\Lambda^*_t$, and let $W:\mathsf M_d\to\mathsf M_d$ be a linear propagator,
i.e. $W\Lambda^*_s=\Lambda^*_t$. Then there exists a {generalized inverse} $B : \M \to \M$ of $\Lambda_s^*$ such that $W = \Lambda^*_t B$ 
%$$    \Lambda^*_s W\, \Lambda^*_s = \Lambda^*_s \ , \quad W = \Lambda^*_t B , $$
if and only if
\[
        \operatorname{Im}W
        \subseteq
        \mathcal S_t,
        \qquad
        \mathcal S_t:=\operatorname{Im}\Lambda^*_t.
\]
\end{proposition}

\begin{proof}
Suppose first that $ W=\Lambda^*_t B$ for some generalized inverse $B$ of $\Lambda^*_s$.  Then, for every $X\in\mathsf M_d$,
\[
        W(X)=\Lambda^*_t(B(X))
        \in\operatorname{Im}\Lambda^*_t.
\]
Consequently,
\[
        \operatorname{Im}W
        \subseteq
        \operatorname{Im}\Lambda^*_t
        =\mathcal S_t.
\]
Conversely, suppose that $\operatorname{Im}W\subseteq\mathcal S_t$. Choose  some generalized inverses ${\Lambda_s^*}^-$ and ${\Lambda^*_t}^-$ of $\Lambda^*_s$ and
$\Lambda^*_t$, respectively:
\[
        \Lambda^*_s{\Lambda_s^*}^-\Lambda^*_s=\Lambda^*_s,
        \qquad
        \Lambda^*_t {\Lambda^*_t}^-\Lambda^*_t=\Lambda^*_t.
\]
Define the corresponding projectors $P_s=\Lambda^*_s{\Lambda_s^*}^-$ and $P_t:=\Lambda^*_t {\Lambda^*_t}^-$. One has 
$ \operatorname{Im}P_s=\mathcal S_s$ and $\operatorname{Im}P_t=\mathcal S_t$. Now define

\begin{equation}
\label{eq:B-detailed}
        B :=  {\Lambda_s^*}^- P_s + {\Lambda^*_t}^- W(\id-P_s).
\end{equation}
We first verify that $B$ is a generalized inverse of $\Lambda^*_s$. One has

\begin{align}
        B\Lambda^*_s
        =
        {\Lambda_s^*}^-P_s\Lambda^*_s
        + {\Lambda^*_t}^- W(\id-P_s)\Lambda^*_s = {\Lambda_s^*}^-\Lambda^*_s ,
\end{align}
due to

\[
        P_s\Lambda^*_s
        =\Lambda^*_s{\Lambda_s^*}^-\Lambda^*_s
        =\Lambda^*_s,
\]
and $(\id-P_s)\Lambda^*_s=0$. Therefore $\Lambda^*_sB\Lambda^*_s = \Lambda^*_s{\Lambda_s^*}^-\Lambda^*_s = \Lambda^*_s$. 
It remains to prove that $\Lambda^*_tB=W$. One computes

\begin{align}
        \Lambda^*_tB &   = \Lambda^*_t \Big( {\Lambda_s^*}^- P_s + {\Lambda^*_t}^- W(\id-P_s) \Big) =
        \Lambda^*_t {\Lambda_s^*}^- P_s  +\Lambda^*_t {\Lambda^*_t}^- W(\id-P_s) \notag\\
        &=
        (W \Lambda^*_s) {\Lambda_s^*}^-  P_s + W(\id-P_s) = W P_s +  W(\id-P_s)=  W ,
\end{align}
where we used $\Lambda^*_t = W \Lambda^*_s$, $P_s^2 = P_s$, and $P_t W=W$ since $P_t$ acts as an identity map on $\im \Lambda^*_t$ (and we assumed that $\operatorname{Im}W\subseteq\mathcal S_t$).  This completes the proof.
\end{proof}

It should be stresses that the generalized inverse constructed above does depend on both $(t,s)$, i.e. $B = B_{t,s}$, meaning that for all pairs $t > s$ there exists $B_{t,s}$ such that 

\[  W_{t,s}^* = \Lambda^*_t B_{t,s} \ , \qquad \Lambda^*_s B_{t,s} \Lambda^*_s = \Lambda^*_s .  \]
The generalized-inverse construction in Proposition~\ref{prop:ginv} requires the stronger statement: there exists ${\Lambda_s^*}^-$ such that for all pairs $t>s$ $W_{t,s}^* = \Lambda^*_t {\Lambda^*_s}^-$,
that is, in this construction ${\Lambda^*_s}^-$ is universal for all $t > s$.

There is a special generalized-inverse construction based on the following

\begin{theorem}[{Composition law from UCP projection}]\label{thm:projection}
Assume Heisenberg CP-divisibility and suppose that, for every $s$,
there {exists} a UCP projection
\[
 P_s:\M\longrightarrow\mathcal S_s,
 \qquad P_s^2=P_s,
 \qquad \im P_s=\mathcal S_s.
\]
Then
\begin{equation}\label{eq:projection-propagator}
 {W^P_{t,s}}^*=w_{t,s}^*P_s
\end{equation}
is UCP, satisfies ${W^P_{t,s}}^*\Lambda^*_s=\Lambda^*_t$, and obeys the
composition law.  Moreover, it can be written as ${W^P_{t,s}}^*=\Lambda^*_t{\Lambda_s^*}^-$ with
one generalized inverse ${\Lambda_s^*}^-$ for each initial time $s$.
\end{theorem}

\begin{proof}
By Heisenberg CP-divisibility, the reduced propagator {$w_{t,s}^*:\mathcal S_s\longrightarrow\mathcal S_t$ is UCP and obeys the composition law \eqref{eq:comp_law_reduced_props}.}
Since \(P_s\) is a projection onto \(\mathcal S_s\), it acts as the
identity on \(\mathcal S_s\). In particular, $P_s\Lambda^*_s=\Lambda^*_s$. It follows that
\[
{W^P_{t,s}}^*:=w_{t,s}^*P_s
\]
is a composition of UCP maps and is therefore
UCP. Furthermore,
\[
{W^P_{t,s}}^*\Lambda^*_s
 =w_{t,s}^*P_s\Lambda^*_s
 =w_{t,s}^*\Lambda^*_s
 =\Lambda^*_t.
\]
We next verify the composition law. Since $w_{u,s}^*P_s(M_d)\subseteq\mathcal S_u$ 
and \(P_u\) acts as the identity on \(\mathcal S_u\), one has
\[
P_u w_{u,s}^*P_s=w_{u,s}^*P_s.
\]
Consequently,
\begin{align*}
{W^P_{t,u}}^*{W^P_{u,s}}^*
 =w_{t,u}^*P_u w_{u,s}^*P_s 
 =w_{t,u}^*w_{u,s}^*P_s 
 =w_{t,s}^*P_s
 ={W^P_{t,s}}^*.
\end{align*}
It remains to establish the generalized-inverse representation.  Choose any algebraic generalized inverse
\(G_s\) of \(\Lambda^*_s\), i.e. $\Lambda^*_s G_s \Lambda^*_s = \Lambda^*_s$, and define  

\[
{\Lambda^*_s}^{-}:=G_sP_s.
\]
Note, that indeed the above prescription defines a generalized inverse of $\Lambda^*_s$

\[
\Lambda^*_s{\Lambda^*_s}^{-} \Lambda^*_s  =\Lambda^*_sG_sP_s \Lambda^*_s = \Lambda^*_sG_s \Lambda^*_s = \Lambda^*_s ,
\]
since $G_s$ is a generalized inverse of $\Lambda^*_s$. 
Finally, using \(\Lambda^*_t=w_{t,s}^*\Lambda^*_s\), we obtain
\[
\Lambda^*_t{\Lambda^*_s}^{-}
 =w_{t,s}^*\Lambda^*_sG_sP_s
 =w_{t,s}^* \Lambda^*_s G_s \Lambda^*_s {\Lambda^*_s}^- 
 =w_{t,s}^* \Lambda^*_s {\Lambda^*_s}^- 
 ={W^P_{t,s}}^*.
\]
The generalized inverse \({\Lambda^*_s}^{-}=G_sP_s\) depends only on the
initial time \(s\), not on \(t\).
\end{proof}

\begin{example}[Phase covariant dynamics]

Consider a Heisenberg picture $\Lambda_t^*$ phase covariant dynamical map \cite{Smirne-ultimate, Teittinen2018, Filippov-ph-cov} which is  continuous, invertible, and CP-divisible for $t< t_*$. Suppose, instead that for $t\ge t_*$

\begin{equation}\label{eq:phase}
 \Lambda^*_t(I)=I,\quad \Lambda^*_t(\sigma_x)=\Lambda^*_t(\sigma_y)=0,\quad
 \Lambda^*_t(\sigma_z)=b_tI+a_t\sigma_z,
\end{equation}
with real $a_t,b_t$. One easily shows that $\{\Lambda^*_t\}_{t \ge t_*}$ is CP if and only if  and $|b_t|+|a_t|\leq 1$. For the corresponding Schr\"odinger picture one finds (for $t \ge t_*$)

%\end{document}

\[
\Lambda_t(I)=I+b_t\sigma_z,\qquad
\Lambda_t(\sigma_x)=\Lambda_t(\sigma_y)=0,\qquad
\Lambda_t(\sigma_z)=a_t\sigma_z.
\]
Assume first that $a_t>0$ for $t \ge t_*$. Then $\mathcal S_t=\operatorname{span}\{I,\sigma_z\}$ and for $t\geq s {\geq t_*}$

\begin{equation}\label{eq:ab}
 w_{t,s}^*(\sigma_z)=\beta_{t,s}I+\alpha_{t,s}\sigma_z,
 \qquad \alpha_{t,s}=\frac{a_t}{a_s},\qquad
 \beta_{t,s}=\frac{b_t-b_s}{a_s}.
\end{equation}
The positivity test on $(I\pm\sigma_z)/2$ for $w_{t,s}^*$ gives the criterion
\begin{equation}\label{eq:Hphase}
 |b_t-b_s|+a_t\leq a_s\ , \qquad t\geq s\geq t_*.
\end{equation}
Since the domain is abelian, the reduced positive map $w^*_{t,s}$ is CP. Composing {the propagator $w_{t,s}$ with the projector $P=E_z$ on $\operatorname{span}\{I,\sigma_z\}$} produces a UCP propagator on all $\mathsf M_2$.
Such propagator is the extension obtained via the projection, as in Theorem \ref{thm:projection} and not necessarily the only UCP extension.

For comparison, the Schr\"odinger propagator $v_{t,s}$ sends {any Bloch vector $(r_x, r_y, r_z)$} to $(0,0,b_t+a_t r_z)$ %\fs{\sout{. The affine action forced on the $z$ axis by a Schr\"odinger propagator has slope $a_t/a_s$ and displacement $b_t-(a_t/a_s)b_s$. Positivity on the two endpoint states is necessary, and a measure-and-prepare channel realizes it whenever} 
and the propagator is positive if it maps $(0,0,\pm1)$ to positive operators, i.e.
\begin{equation}\label{eq:Sphase}
 \left|b_t-\frac{a_t}{a_s}b_s\right|+\frac{a_t}{a_s}\leq1.
\end{equation}
Equations~\eqref{eq:Hphase} and~\eqref{eq:Sphase} are therefore Heisenberg and Schr\"odinger  CP-divisibility criteria, respectively, {for times $t\geq s\geq t_*$ for which the dynamical map is non-invertible.}

If $a_s=0$, then $\Lambda^*_s(X)=\tr(\rho_sX)I$ for a fixed state $\rho_s {= (I + b_s \sigma_z)/2}$. Any unital propagator fixes scalar matrices, so Heisenberg divisibility forces $\Lambda^*_t=\Lambda^*_s$ for all later times. In this phase-covariant representation this means $a_t=0$ and $b_t=b_s$.

%Appendix~\ref{app:singular-qubits} records the corresponding
%generalized-inverse calculations, including the full Choi conditions
%and the rank-three anisotropic collapse.

\end{example}

\begin{remark}
Let $\mathcal B\subseteq\M$ be a unital $C^*$-subalgebra.  A
conditional expectation onto $\mathcal B$ is a positive norm-one
projection $E_{\mathcal B}:\M\to\mathcal B$.  It satisfies the
bimodule identity
\begin{equation}\label{eq:bimodule}
 E_{\mathcal B}(B_1XB_2)
 =B_1E_{\mathcal B}(X)B_2,
 \qquad B_1,B_2\in\mathcal B.
\end{equation}
Tomiyama's theorem~\cite{Tomiyama} states, conversely, that every
norm-one linear projection from a $C^*$-algebra onto a
$C^*$-subalgebra is positive and has the bimodule
property~\eqref{eq:bimodule}; it is therefore a conditional
expectation.  Conditional expectations are completely positive.  In
particular, a unital positive projection onto $\mathcal B$ has norm
one and hence is a conditional expectation.  For a matrix
subalgebra, the orthogonal projection with respect to the
Hilbert--Schmidt inner product is the normalized-trace-preserving
conditional expectation and is therefore UCP.  This is the canonical
projection that may be used in Theorem~\ref{thm:projection}.  An
arbitrary linear projection onto the same subalgebra need not be
positive or completely positive.

Tomiyama's theorem cannot be invoked when the range of $\Lambda^*_s$ is merely an
operator system and is not closed under the ambient matrix product in $\M$.  This
does not, however, rule out a UCP projection.  The Choi--Effros
theorem~\cite{ChoiEffros} says that if $P:\M\to\M$ is a completely
positive contractive projection, then $\mathcal S=\im P$ becomes a
$C^*$-algebra with its inherited involution and norm but with the new
product
\begin{equation}\label{eq:CEproduct}
 X\odot Y:=P(XY),\qquad X,Y\in\mathcal S.
\end{equation}
In general $X \odot Y\neq XY$.  Hence the range of a UCP projection is
always a $C^*$-algebra in the Choi--Effros sense, but need not be a
$C^*$-subalgebra of the ambient matrix algebra. Recently, the Choi--Effros product was used in the analysis of asymptotic dynamics in the Heisenberg picture \cite{Amato}.
    
\end{remark}

One might expect that Heisenberg CP-divisibility implies the
existence, for every \(s\), of a    generalized inverse \({\Lambda_s^*}^-\) of
\(\Lambda^*_s\) such that $ {W_{t,s}^G}^*:=\Lambda^*_t{\Lambda_s^*}^-$ is completely positive for all \(t\geq s\). This is not true (cf. Appendix \ref{A-1} for a counterexample with $d=3$). However, as we show in the next Section, it is always true for $d=2$. 

%%%%%%%%%%%%%%%%%%%%%%%%%%%%%%%%%%%%%%%%%%%%%%%%%%%%%%%%%%%%%%%%%%%%%%
\subsection{Why qubit dynamics is special:
Schr\"odinger versus Heisenberg divisibility}
\label{sec:qubit-special}

%\fs{\it Subsection as well}

The qubit case plays a special role in the relation between
information-flow conditions and divisibility.  However, the reason is
different in the Schr\"odinger and Heisenberg pictures. 

%\subsection{The Schr\"odinger-picture}

Consider a CPTP dynamical map $\{\Lambda_t\}_{t \ge 0}$. Schr\"odinger CP-divisibility requires CPTP propagators satisfying $\Lambda_t=V_{t,s}\Lambda_s$ for $ t \geq s$. 
For noninvertible maps, the reduced propagator $v_{t,s}$ is determined
only on $\operatorname{Im}\Lambda_s$. Even if $v_{t,s}$ is CPTP the CP extension $V_{t,s}$ might not automatically be trace preserving on the full algebra $\M$ \cite{CRS}. For qubit dynamics it has been proven \cite{chakraborty_information_2019}  that this obstruction disappears.  More
precisely, a qubit dynamical map is Schr\"odinger CP-divisible if and
only if
\begin{equation}
\label{eq:chakraborty_information_2019-information-condition}
 \frac{d}{dt}
 \left\|
   (\operatorname{id}_2\otimes\Lambda_t)
   (p_1\rho_1-p_2\rho_2)
 \right\|_1
 \leq0
\end{equation}
for all probability distributions $p_1+p_2=1$ and all two-qubit
density matrices $\rho_1,\rho_2$. Equivalently,

\[ \frac{d}{dt}
 \left\|
   (\operatorname{id}_2\otimes\Lambda_t)
   (X)
 \right\|_1
 \leq0 ,  \]
for all $X = X^* \in \mathsf M_2 \otimes \mathsf M_2 . $
Two specifically qubit ingredients enter the proof.  First, there is
no CPTP projection onto a three-dimensional subspace of
$\mathsf M_2$ spanned by density matrices \cite{chakraborty_information_2019, braun_universal_2014}.  Hence, at any given time,
the image of a qubit dynamical map satisfying
Eq.~\eqref{eq:chakraborty_information_2019-information-condition} can have dimension only 1, 2, or 4. 
Second, when the image is two-dimensional, two density matrices
suffice to span its affine part, and the Alberti--Uhlmann theorem
provides a CPTP extension of the reduced propagator.  This mechanism
is special to qubits and does not extend directly to arbitrary
dimension.

%\subsection{Why the Heisenberg norm criterion is not qubit-specific}

For the corresponding Heisenberg problem, the analogue of the complete
information-flow characterization holds in every finite dimension.
Indeed, Theorem~\ref{thm:main} states that Heisenberg CP-divisibility
is equivalent to
\begin{equation}
\label{eq:Heisenberg-complete-monotonicity}
 \left\|
   (\operatorname{id}_d\otimes\Lambda^*_t)(X)
 \right\|_\infty
 \leq
 \left\|
   (\operatorname{id}_d\otimes\Lambda^*_s)(X)
 \right\|_\infty
\end{equation}
for every Hermitian
$X\in\mathsf M_d\otimes\mathsf M_d$ and every $t\geq s$, due to the fact that the reduced UCP propagator $ w_{t,s}^*:\mathcal S_s\longrightarrow\mathcal S_t$ can be extended to a UCP map $ W_{t,s}^*:\mathsf M_d\longrightarrow\mathsf M_d$ by means of Arveson's theorem \cite{Arveson}.
No qubit geometry and no analogue of the
Alberti--Uhlmann theorem are required. Consequently, the qubit is not special for the Heisenberg
operator-norm characterization of CP-divisibility.  It is,
nevertheless,  for generalized-inverse propagators.

%fs{\it Move the Lemma right after the first sentence of the Section? The non-existence of CP projection is %used also for Schr\"odinger.}

\begin{lemma}
\label{lem:no-three-dimensional-UCP-projection}
Let $\mathcal S\subset\mathsf M_2$ be a three-dimensional operator
system.  There is no UCP projection
\[
        P:\mathsf M_2\longrightarrow\mathcal S.
\]
\end{lemma}

\begin{proof}
Every 3-dimensional operator system in $\mathsf M_2$ is unitarily equivalent to $\operatorname{span}\{I,\sigma_x,\sigma_z\}.$ 
Suppose that a UCP projection $P$ onto $\mathcal S$ existed.  Since a
projection acts as the identity on its range, one has
\[
        P(\sigma_x)=\sigma_x,
        \qquad
        P(\sigma_z)=\sigma_z.
\]

Let $\mathcal A$ and $\mathcal B$ be unital $C^*$-algebras and let
$\Phi:\mathcal A\to\mathcal B$ be UCP. The multiplicative domain of 
 $\Phi$ is defined by \cite{Choi1974,Paulsen,Stormer_book} (see also recent papers \cite{MD1,MD2})
\[
\operatorname{MD}(\Phi)
=
\left\{
a\in\mathcal A:
\Phi(a^* a)=\Phi(a^*) \Phi(a),\quad
\Phi(aa^*)=\Phi(a)\Phi(a^*) 
\right\}.
\]
 In particular,
$\operatorname{MD}(\Phi)$ is the largest $C^*$-subalgebra of
$\mathcal A$ on which $\Phi$ acts as a $*$-homomorphism.

Now, both $\sigma_x$ and $\sigma_z$ are unitaries {preserved} by $P$ and hence
\[
 P(\sigma_k \sigma_k)
 =
 P(\sigma_k)  P(\sigma_k)
 =
 I,
 \qquad k=x,z.
\]
Hence $\sigma_x, \sigma_z\in\operatorname{MD}(P)$.  It follows that
\[
\begin{aligned}
 P(\sigma_x\sigma_z) =
 P(\sigma_x)P(\sigma_z) =
 \sigma_x\sigma_z
 =
 -i\sigma_y.
\end{aligned}
\]
This is impossible because $\sigma_y\notin\mathcal S$, whereas $\operatorname{Im}P=\mathcal S$.
\end{proof}

Notice that the conclusion is specific to complete positivity  (actually, Schwarz property $\Phi(a^*a) \geq \Phi(a^*)\Phi(a)$ is sufficient). There does exist a positive unital projection onto
$\operatorname{span}\{I,\sigma_x,\sigma_z\}$, namely
\[
 P(I)=I,\qquad
 P(\sigma_x)=\sigma_x,\qquad
 P(\sigma_y)=0,\qquad
 P(\sigma_z)=\sigma_z.
\]
It deletes one component of the Bloch vector and is therefore
positive, but Lemma~\ref{lem:no-three-dimensional-UCP-projection}
shows that it cannot be completely positive.

The absence of three-dimensional UCP projections has a strong
dynamical consequence.

\begin{theorem}
\label{thm:qubit-Heisenberg-ranks}
Let
\[
        \Lambda^*_t:\mathsf M_2\longrightarrow\mathsf M_2,
        \qquad t\geq0,
\]
be a norm-continuous family of UCP maps such that $\Lambda^*_{t=0}=\operatorname{id}$. 
If the family is Heisenberg CP-divisible, then, at every finite time,
\[
        \operatorname{rank}\Lambda^*_t\in\{1,2,4\}.
\]
In particular, a continuous Heisenberg CP-divisible qubit evolution
cannot have a three-dimensional accessible operator system.
\end{theorem}

\begin{proof}
Heisenberg divisibility implies
\[
        \ker\Lambda^*_s\subseteq\ker\Lambda^*_t,
        \qquad t\geq s,
\]
and therefore $t\mapsto\operatorname{rank}\Lambda^*_t$ is
non-increasing. Suppose, contrary to the claim, that $\operatorname{rank}\Lambda^*_{t_1}=3$ 
for some finite $t_1$.  Let
\[
 t_*:=
 \inf\bigl\{
   t\in[0,t_1]:
   \operatorname{rank}\Lambda^*_t<4
 \bigr\}
\]
be the first singular time.  Since $\Lambda^*_0=\operatorname{id}$ and the
family is continuous, one has $t_*>0$.  Rank monotonicity and the
assumption $\operatorname{rank}\Lambda^*_{t_1}=3$ imply
\[
        \operatorname{rank}\Lambda^*_{t_*}=3.
\]

For every $s<t_*$, the map $\Lambda^*_s$ is invertible.  Hence the
propagator from $s$ to $t_*$ is uniquely determined:
\[
        W_{t_*,s}^*
        =
        \Lambda^*_{t_*}{\Lambda^*_s}^{-1}.
\]
By Heisenberg CP-divisibility, $W_{t_*,s}^*$ is UCP.  Moreover,
\[
        \operatorname{Im}W_{t_*,s}^*
        \subseteq {\mathcal S_{t_*},\qquad\mathcal S_{t_*}\coloneqq}
        \operatorname{Im}\Lambda^*_{t_*}.
\]
Now, choose a sequence $s_n\uparrow t_*$.  The set of UCP maps on
$\mathsf M_2$ is compact, so, $ W_{t_*,s_n}\longrightarrow P$ for some UCP map $P$.  Since
\[
        W_{t_*,s_n}\Lambda^*_{s_n}=\Lambda^*_{t_*}
\]
and $\Lambda^*_{s_n}\to\Lambda^*_{t_*}$, taking the limit gives $ P\Lambda^*_{t_*}=\Lambda^*_{t_*}$.
Furthermore, one also has
\[
        \operatorname{Im}P\subseteq\mathcal S_{t_*}.
\]
If $Y\in\mathcal S_{t_*}$, then $Y=\Lambda^*_{t_*}(X)$ for some $X {\in\mathsf M_2}$, and
therefore
\[
        P(Y)
        =
        P\Lambda^*_{t_*}(X)
        =
        \Lambda^*_{t_*}(X)
        =
        Y.
\]
Together with
$\operatorname{Im}P\subseteq\mathcal S_{t_*}$, this implies
\[
        \operatorname{Im}P=\mathcal S_{t_*},
        \qquad
        P^2=P.
\]
Hence $P$ is a UCP projection onto the three-dimensional operator
system $\mathcal S_{t_*}$, contradicting
Lemma~\ref{lem:no-three-dimensional-UCP-projection}.
\end{proof}

The preceding theorem implies that the generalized-inverse
construction is particularly well behaved for continuous qubit
dynamics.

\begin{corollary}
\label{cor:qubit-generalized-inverse}
Under the assumptions of
Theorem~\ref{thm:qubit-Heisenberg-ranks}, for every $s\geq0$ there
exists a    generalized inverse ${\Lambda_s^*}^-$ of $\Lambda^*_s$ such that
\[
        {W_{t,s}^G}^*:=\Lambda^*_t{\Lambda_s^*}^-
\]
is UCP for every $t\geq s$.  Moreover,
\[
        {W^G_{t,u}}^*{W^G_{u,s}}^*={W_{t,s}^G}^*,
        \qquad t\geq u\geq s.
\]
\end{corollary}

\begin{proof}
By Theorem~\ref{thm:qubit-Heisenberg-ranks}, one has $ \dim\mathcal S_s\in\{1,2,4\}$. 
Each of these qubit operator systems is the range of a UCP
projection. Indeed, if $\dim\mathcal S_s=4$, then the map $\Phi_s$ is invertible. If $\dim\mathcal S_s=1$, then
$ \mathcal S_s=\mathbb C I$, then for example $P_s(X)= \frac 12 {\rm Tr}X\,I$ is a UCP projection onto $\mathcal S_s$. Finally, suppose that $\dim\mathcal S_s=2$ and $ \mathcal S_s=\operatorname{span}\{I,A\}$
for a non-scalar Hermitian matrix $A$. Note that $\operatorname{span}\{I,A\} = \operatorname{span}\{I,U\}$, where $U$ is a unitary defined by

\[
U:=\frac{2A- ({\rm Tr}\, A)\, I}
         {\lambda_{+}-\lambda_{-}} ,
\]
with $\lambda_+> \lambda_-$ being real eigenvalues of $A$.   One may take $ P_s(X)=\frac12\bigl(X+UXU\bigr)$. 
and hence by Theorem \ref{thm:projection} the propagator ${W^P_{t,s}}^*\Lambda^*_s=\Lambda^*_t$ can be written as $\Lambda^*_t{\Lambda_s^*}^-$ with one generalized inverse ${\Lambda_s^*}^-$ for each initial time $s$.
\end{proof}

Thus, for continuous Heisenberg CP-divisible qubit dynamics starting
from the identity, the conjecture that one can choose completely
positive generalized-inverse propagators is true. This also explains why the continuous counterexample in
Proposition~\ref{prop:continuous-counterexample} requires dimension
three rather than a qubit.

\section{Time-local generators}
\label{sec:tl_gen}

%Consider now time-local master equations for dynamical maps $\{\Lambda_t\}_{t\ge 0}$.  
Assuming that an invertible dynamical map $\{\Lambda_t\}_{t\ge 0}$ is differentiable  it satisfies time-local master equation with left $\mathcal{L}_t$ and right $\mathcal{R}_t$ regular time-local generators \cite{Settimo-SchroHeis}

\begin{equation}
    \dot{\Lambda}_t = \mathcal{L}_t \Lambda_t  \ , \qquad  \dot{\Lambda}_t = \Lambda_t \mathcal{R}_t .
     \label{eq:left-right-ME}
\end{equation}
Taking adjoints in Eq.~\eqref{eq:left-right-ME} yields in the Heisenberg picture
\begin{equation}
 \dot\Lambda^*_t
 =\Lambda^*_t  \mathcal L_t^* \ , \qquad \dot{\Lambda}_t^* 
 = \mathcal R_t^* \Lambda^*_t.
 \label{eq:dual-left-right-generators}
\end{equation}
In particular, the Heisenberg left generator is $\mathcal R_t^*$, not $\mathcal L_t^*$. 
%If the dynamical map $\{\Lambda_t\}_{t\ge 0}$ 
Note that given the map $\{\Lambda_t\}_{t \ge 0}$   the left and right generators are uniquely defined via \cite{Settimo-SchroHeis}
\begin{equation}
 \mathcal L_t:=\dot\Lambda_t \Lambda_t^{-1},
 \qquad
 \mathcal R_t:=\Lambda_t^{-1} \dot\Lambda_t ,
 \label{eq:left-right-generators}
\end{equation}
and they are related by the similarity transformation

\begin{equation}
% \dot\Lambda_t= \mathcal L_t \Lambda_t
% =\Lambda_t   \mathcal R_t,
% \qquad
 \mathcal R_t=\Lambda_t^{-1}  \mathcal L_t \Lambda_t.
 \label{eq:left-right-master-equations}
\end{equation}
Accordingly, one finds for the corresponding propagators
\begin{equation}
 V_{t,s}
 =\mathcal T_{\leftarrow}
   \exp\!\left(\int_s^t \mathcal L_\tau\,d\tau\right),
 \qquad
 W_{t,s}
 =\mathcal T_{\rightarrow}
   \exp\!\left(\int_s^t \mathcal R_\tau\,d\tau\right),
 \label{eq:left-right-ordered-exponentials}
\end{equation}
where $\mathcal T_{\leftarrow}$ and $\mathcal T_{\rightarrow}$ denote
chronological and antichronological ordering, respectively.

%\begin{equation}
%\mathcal K_t^{\mathrm L}:=\dot\Lambda^*_t {\Lambda^*_t}^{-1}= \mathcal R_t^*,
% \qquad
%\mathcal K_t^{\mathrm R}:={\Lambda^*_t}^{-1} \dot\Lambda^*_t= \mathcal L_t^*.
% \label{eq:Heis-left-right-generators}
%\end{equation}

Under the usual regularity assumptions for time-local master equations,
Schr\"odinger CP-divisibility is equivalent to $\mathcal L_t$ having Gorini-Kossakowski-Lindblad-Sudarshan (GKLS) form with nonnegative rates $\gamma_{i,t}$ for almost every $t$ \cite{GKS, L}

\begin{equation}
    \mathcal{L}_t(\rho) = -i[H_t,\rho] + \sum_i \gamma_{i,t} \Big( L_{i,t} \rho L_{i,t}^* - \frac 12 \{ L_{i,t}^* L_{i,t},\rho\} \Big) .
\end{equation}
Heisenberg CP-divisibility
is instead equivalent to $\mathcal R_t$ having GKLS form with nonnegative rates,
or, equivalently, to $\mathcal R_t^*$ generating UCP intermediate maps.  The
similarity relation in Eq.~\eqref{eq:left-right-master-equations} does
not preserve the cone of GKLS generators.  This is the generator-level
origin of the inequivalence established in Ref.~\cite{Settimo-SchroHeis}.

Consider now $\{\Lambda_t\}_{t \ge 0}$ which might be non-invertible. 

\begin{proposition}
%If $\Lambda_t$ is right-differentiable but in general  not invertible, then 
There exists a time-local generator $\mathcal{L}_t: \M \to \M$ such that $ \dot{\Lambda}_t = \mathcal{L}_t \Lambda_t$  if and only if 

\begin{equation}\label{eq:derivativekernel}
 \ker\Lambda_t\subseteq\ker\dot\Lambda_t.
\end{equation}
    
\end{proposition}
\begin{proof} Necessity is immediate. Now, let $\Lambda_t^-$ be a generalized inverse of $\Lambda_t$. Since $\Lambda_t\bigl(\operatorname{id}-\Lambda^-_t\Lambda_t\bigr)=0$,  assuming (\ref{eq:derivativekernel}) one finds for any $X \in \M$

\[
\bigl(\operatorname{id}-\Lambda^-_t\Lambda_t\bigr)(X)
\subseteq\ker\Phi_t
\subseteq\ker\dot\Phi_t,
\]
and hence $\dot{\Lambda}_t \bigl(\operatorname{id}-\Lambda^-_t\Lambda_t\bigr) =0$, that is, $\dot{\Lambda}_t  = (\dot{\Lambda}_t \Lambda_t^-) \Lambda_t$. 
Consequently, all solutions of $ \dot{\Lambda}_t = \mathcal{L}_t \Lambda_t$  are given by
\[
\mathcal L_t 
=
\dot\Lambda_t \Lambda^-_t
+
Z_t\bigl(\operatorname{id}-\Lambda_t \Lambda^-_t\bigr),
\]
where $Z_t: \M\to \M$ is arbitrary.
\end{proof}
Similarly, one finds for the left Heisenberg generator  $\mathcal{R}^*_t$ if and only if

\[ \ker\Lambda^*_t\subseteq\ker\dot\Lambda_t^* , \]
which has the following form 

\begin{equation}\label{eq:generator}
 \mathcal R^*_t=\dot\Lambda^*_t {\Lambda^*_t}^- + Y_t^*(\id-\Lambda^*_t {\Lambda^*_t}^-),\qquad
 \dot\Lambda^*_t=\mathcal R^*_t\Lambda^*_t,
\end{equation}
with arbitrary $Y_t : \M \to \M$. This is the generator counterpart of~\eqref{eq:all}. Its action on $\mathcal S_t$ is fixed, while its action on a complement is not. It satisfies $\mathcal R^*_t(I)=0$. One finds the following relations between $\mathcal{L}_t$ and $\mathcal{R}_t$:

\begin{equation}
    \mathcal L_t 
=
\Lambda_t \mathcal{R}_t \Lambda^-_t
+
Z_t\bigl(\operatorname{id}-\Lambda_t \Lambda^-_t\bigr), \qquad 
 \mathcal R_t= \Lambda^-_t \mathcal{ L}_t \Lambda_t + (\id-\Lambda^-_t {\Lambda_t})Y_t ,
\end{equation}
which provide the generalization of the similarity relation (\ref{eq:left-right-master-equations}).

\begin{example} Consider a qubit CPTP family which is for $t \in [0,\frac 43]$ is defined as follows

\[  \Lambda_t(\rho) = (1-t) \rho + \frac t2 I\, {\rm Tr}\rho . \]
At $t=1$ one has $\Lambda_{t=1}(\rho) =  \frac 12 I\, {\rm Tr}\rho$ which has a 3-dimensional kernel spanned by traceless matrices $\ker \Lambda_{t=1} = {\rm span}\{\sigma_x,\sigma_y,\sigma_z\}$. Now $\dot{\Lambda}_t(\rho) =  \frac 12 I\, {\rm Tr}\rho- \rho$, and hence $\ker \dot{\Lambda}_{t=1} = cI$.  The kernel condition $\ker\Lambda_t\subseteq\ker\dot\Lambda_t$ is violated  and hence there is no time-local generator at $t=1$.
{Notice that $\Lambda^*_t = \Lambda_t$ and therefore the condition $\ker\Lambda_t^*\subseteq\ker\dot\Lambda_t^*$ is violated also in the Heisenberg picture, and therefore also the time-local generator $\mathcal R^*_t$ does not exist.}
\end{example}

The next example illustrates a scenario where $\Lambda_t^* = \Lambda_t$ is not invertible but Schr\"odinger and Heisenberg CP-divisible and the generator $\mathcal{L}_t = \mathcal{R}_t$ does exist for all $t \geq 0$. 

\begin{example}[Non-invertible CP-divisible evolution and its generator]
%\fs{\it The only simplification might be to not explicitly specify $f(t)$, but any monotonic, positive, continuous and differentiable function with $f(0) = 1$, $\dot f(1) = 0$, $f(t) = 0$ for $t\ge1$ works. Divergence of the rate is guaranteed by $\int_0^1\kappa(t) dt = \log f(1)-\log f(0) = +\infty$. But it doesn't seem to be a huge simplification, if any.}

Consider the qubit dephasing channels
\[
\mathcal D_q(X)
=
\frac{1+q}{2}X
+\frac{1-q}{2}\sigma_zX\sigma_z,
\qquad 0\leq q\leq 1.
\]
They satisfy $\mathcal D_q\mathcal D_r=\mathcal D_{qr}$ and $\mathcal D_q^*=\mathcal D_q$.
Let
\[
f(t)=
\begin{cases}
(1-t)^2, & 0\leq t<1,\\
0, & t\geq 1,
\end{cases}
\]
and define a dynamical map $\Lambda_t=\mathcal D_{f(t)}$.
The function $f$ is continuously differentiable, and hence locally
absolutely continuous.  In the Pauli basis one has
\[
\Lambda_t(I)=I,
\qquad
\Lambda_t(\sigma_z)=\sigma_z,
\qquad
\Lambda_t(\sigma_x)=f(t)\sigma_x,
\qquad
\Lambda_t(\sigma_y)=f(t)\sigma_y.
\]
Consequently,
\[
\operatorname{rank}\Lambda_t=
\begin{cases}
4, & t<1,\\
2, & t\geq 1.
\end{cases}
\]
For $t\geq1$ one has 
\[
\Lambda_t=\mathcal E_z,
\qquad
\mathcal E_z(X)=\frac12\bigl(X+\sigma_zX\sigma_z\bigr),
\]
where the conditional expectation $\mathcal E_z$ is the complete dephasing channel. The evolution is CP-divisible.  Indeed, define
\[
V_{t,s}=
\begin{cases}
\displaystyle
\mathcal D_{f(t)/f(s)}, & 0\leq s\leq t<1,\\[2mm]
\mathcal D_0=\mathcal E_z, & 0\leq s<1\leq t,\\
\mathrm{id}, & 1\leq s\leq t.
\end{cases}
\]
%{alternative:
%\[
%V_{t,s}=
%\begin{cases}
%\displaystyle
%\mathcal D_{f(t)/f(s)}, & t<1,\\[2mm]
%\mathcal D_0=\mathcal E_z, & t \geq 1,\\
%\end{cases}
%\]
%and $0\leq s\leq t$.
%}
Every $V_{t,s}$ is CPTP and $\Lambda_t=V_{t,s}\Lambda_s$. Moreover, the identity $\mathcal D_q\mathcal D_r=\mathcal D_{qr}$ implies the composition rule $V_{t,u}V_{u,s}=V_{t,s}$ for $t\geq u\geq s$.
Since {$\Lambda_t^*=\Lambda_t$} is self-adjoint, one has $W_{t,s}=V_{t,s}$, and the evolution is CP-divisible in both the Schr\"odinger and the Heisenberg pictures.

To analyze the time-local generator, {observe that $\Lambda_t$ can be written as}
%introduce $\mathcal Q:=\mathrm{id}-\mathcal E_z$. Then
\[\Lambda_t=\mathcal E_z+f(t){(\operatorname{id} - \mathcal E_z).}\]
For $t<1$, it is invertible, with
\[
\Lambda_t^{-1}
=
\mathcal E_z+\frac{1}{f(t)} {(\operatorname{id} - \mathcal E_z).}
\]
{Furthermore,} its left and right generators coincide:
%\[
%\mathcal L_t
%=
%\dot\Lambda_t\Lambda_t^{-1}
%=
%\frac{\dot f(t)}{f(t)}\mathcal Q
%=
%\Lambda_t^{-1}\dot\Lambda_t
%=
%\mathcal R_t.
%\]
%Since
%\[\frac{\dot f(t)}{f(t)}=-\frac{1}{1-t},\]
%\fs{Therefore,} one obtains for $t < 1$
\[
{
\mathcal L_t(X)=\mathcal R_t(X)
=
\kappa(t) \bigl(\sigma_zX\sigma_z-X\bigr),
\qquad \kappa(t)={\frac{\dot f(t)}{f(t)}=}\frac{1}{1-t}\geq0.
}
\]
This is a GKLS generator.  The rate $\kappa(t)$ is not integrable up to the rank-loss time:
\[
\int_0^1\kappa(t)\,\mathrm dt
=
\int_0^1\frac{\mathrm dt}{1-t}\,
=
+\infty.
\]
Thus the generator is of GKLS form for every $t<1$, but it is not
integrable on an interval containing $t=1$.  This divergence is
precisely what allows the dynamical map to lose invertibility at a
finite time $t=1$. At and after the singular time,
\[
\dot\Lambda_t=0,
\qquad t\geq1.
\]
Since $\mathcal E_z$ is a projection and hence an inner generalized
inverse of itself, all solutions of the left-generator equation $\dot\Lambda_t=\mathcal L_t\Lambda_t$
are of the form
\[
\mathcal L_t
=
Z_t\bigl(\mathrm{id}-\mathcal E_z\bigr),
\qquad t\geq1,
\]
where $Z_t$ is arbitrary.  Similarly, all solutions of the
right-generator equation $\dot\Lambda_t=\Lambda_t\mathcal R_t$ are
\[
\mathcal R_t
=
\bigl(\mathrm{id}-\mathcal E_z\bigr)Y_t,
\qquad t\geq1.
\]
the simplest choice for $t \geq 1$ is $\mathcal L_t=\mathcal R_t=0$. 
There are, however, infinitely many nonzero GKLS continuations.  For
example,
\[
\mathcal L_t(X)
=
c(t)\bigl(\sigma_zX\sigma_z-X\bigr),
\qquad c(t)\geq0,
\]
is admissible because it vanishes on $\operatorname{Im}\mathcal E_z
=
\operatorname{span}\{I,\sigma_z\}$. 
This example therefore illustrates that CP-divisibility remains well
defined through a finite-time loss of invertibility.  Before the
singular time the generator is unique and necessarily diverges,
whereas at and after the singular time it is highly nonunique because
the evolution fixes its action only on the accessible range.

\end{example}

The arbitrary maps $Z_t$ and $Y_t$ describe a natural gauge freedom of the time-local generator associated with the non-invertibility of $\Lambda_t$. Indeed, the term $Z_t(\mathrm{id}-\Lambda_t\Lambda^-_t)$ vanishes on $\operatorname{Im}\Lambda_t$ and therefore does not affect $\mathcal L_t\Lambda_t$, whereas $(\mathrm{id}-\Lambda^-_t\Lambda_t)Y_t$ takes values in $\ker\Lambda_t$ and hence does not affect $\Lambda_t\mathcal R_t$. Consequently, different choices of $Z_t$ and $Y_t$ define different extensions of the generator to dynamically inaccessible directions while producing exactly the same evolution $\dot{\Lambda}_t$. This freedom disappears when $\Lambda_t$ is invertible. It should, however, be understood as an algebraic gauge freedom: additional requirements such as Hermiticity preservation, trace preservation, unitality, regularity, or GKLS form may substantially restrict the admissible choices of $Z_t$ and $Y_t$.

To move to the gauge-invariant formulation, assume $\ker\Lambda_t\subseteq\ker\dot\Lambda_t$ and define

\[
  \widehat{\Lambda}_t:
  \mathrm{M}_d/\ker\Lambda_t\longrightarrow\operatorname{Im}\Lambda_t,
  \qquad
  \widehat{\Lambda}_t([X])=\Lambda_t(X),
\]
together with the induced map
\[
  \widehat{\dot\Lambda}_t([X])=\dot\Lambda_t(X).
\]
Then the gauge-invariant left generator at the fixed time $t$ is
\[
  \ell_t
  =\widehat{\dot\Lambda}_t\,\widehat{\Lambda}_t^{-1}
  :\operatorname{Im}\Lambda_t\longrightarrow\mathrm{M}_d.
\]
The analogous pointwise construction for $\Lambda_t^*$ gives
\[
  r_t^*
  =\widehat{\dot\Lambda_t^*}\,
   \widehat{\Lambda_t^*}^{-1}
  :\operatorname{Im}\Lambda_t^*\longrightarrow\mathrm{M}_d,
\]
If the kernels are constant on a time interval (so that the quotient spaces
can be identified with one fixed space), then one may write
\[
  \frac{d}{dt}\widehat{\Lambda}_t=\ell_t\widehat{\Lambda}_t,
  \qquad
  \frac{d}{dt}\widehat{\Lambda_t^*}
  =r_t^*\widehat{\Lambda_t^*}.
\]
Across a rank-changing time these differential equations require a one-sided
or almost-everywhere interpretation; without a chosen identification of the
varying quotient spaces, $d\widehat\Lambda_t/dt$ is only formal.

\section{Conclusions}
\label{sec:conclusions}

We have developed a formulation of Heisenberg divisibility that remains valid
when the dynamical maps cease to be invertible.  The basic object is then not a
unique map on the full matrix algebra, but the unique reduced propagator on the
accessible operator system $\mathcal{S}_s=\operatorname{Im}\Lambda_s^*$.
Heisenberg divisibility is governed by the inclusion of the kernels of
the adjoint maps, while complete-positive divisibility is equivalent to
complete operator-norm monotonicity (\ref{eq:complete-norm}).  This gives a
finite-time characterization that does not require differentiability and
continues through changes of rank.

The distinction between a reduced propagator $w^*_{t,s}$ and its extension $W^*_{t,s}$ is central.
Arveson's theorem guarantees a UCP extension of each reduced UCP propagator,
but it does guarantee that $W^*_{t,s}$ satisfy local composition law.  
Generalized inverses solve this  problem: they produce unital propagators $\Lambda_t^* \Lambda^{*-}_s$ with an exact composition law.  They need not,
however, be completely positive.
A UCP projection onto each accessible operator system $\mathcal{S}_s$  reconciles the two
requirements and yields a composable UCP family.  When  $\mathcal{S}_s$  is a
$C^*$-subalgebra, this construction is supplied by a conditional expectation
and Tomiyama's theorem.  For a general operator-system range, the appropriate
intrinsic algebraic structure is instead the Choi--Effros product. It should be stressed that in the Schr\"odinger picture it was proved  \cite{CRS} that CP-divisibility guarantees the existence of reduced CPTP propagator $v_{t,s}$ on $\im \Lambda_s$ which can extended to CP maps $V_{t,s}$ which in general are not trace-preserving.  

The qutrit counterexample shows that the projection requirement is genuinely
stronger than Heisenberg CP-divisibility: even a continuous, locally absolutely
continuous, CP-divisible evolution starting from the identity can reach an
accessible operator system  $\mathcal{S}_s$  for which no completely positive
generalized-inverse propagator exists.  In contrast, continuous qubit
dynamics cannot pass from full rank to a three-dimensional accessible operator
system under Heisenberg CP-divisibility.  The remaining possible ranges have
dimensions $1$, $2$, or $4$ and admit UCP projections.  This explains why the
generalized-inverse conjecture is valid for qubits but fails already in
dimension three.

At the differential level, loss of invertibility replaces a unique time-local
generator by an affine family of generators.  Their common restriction to the
accessible range  $\mathcal{S}_s$  is fixed by the dynamics, whereas their action on inaccessible
directions is a gauge freedom.   The quotient
and range formulation isolates the gauge-invariant part of this description.

Several questions remain open.  It would be useful to characterize precisely
which higher-dimensional accessible operator systems are UCP-complemented,
and to find necessary and sufficient conditions for the existence of a global
UCP cocycle when no such projections are available.  A second direction is to
relate these two-time results to operational tasks (cf. \cite{Federico_Tasks})
and to genuinely multi-time notions of quantum Markovianity (cf. \cite{Pollock2018}).

\section*{Acknowledgements}
D.C. was supported by the Polish National Science
Centre project No. 2024/55/B/ST2/01781.
F.S. acknowledges support from Magnus Ehrnroothin S\"a\"ati\"o.

\appendix

\section*{Appendix}

\numberwithin{equation}{section}

\section{CP-divisibility does not imply a completely positive
generalized-inverse propagator}   \label{A-1}

{In this section, we show that there exist Heisenberg CP-divisible dynamics on $\mathsf M_3$ for which $\Lambda_*{\Lambda_s^*}^-$ is not CP for any generalized inverses ${\Lambda_s^*}^-$.
We prove this in Proposition~\ref{prop:continuous-counterexample}, whose proof relies on the following Lemma}

\begin{lemma}\label{lem:three-dimensional-range}
Let $\mathcal S\subset\mathsf M_3$ be a three-dimensional operator system.
If there exists a UCP projection
\[
        P:\mathsf M_3\longrightarrow\mathcal S,
\]
then $\mathcal S$ is a commutative $C^*$-subalgebra of $\mathsf M_3$.
\end{lemma}

\begin{proof}
By the Choi--Effros theorem~\cite{ChoiEffros}, the range of $P$, equipped with the
product
\[
        X \odot  Y:=P(XY)
\]
{as in Eq.~\eqref{eq:CEproduct},}
is a finite-dimensional \(C^*\)-algebra whose matrix order is the
matrix order inherited from \(M_3\). Since $\dim_{\mathbb C}\mathcal S=3$,
this \(C^*\)-algebra is isomorphic to \(\mathbb C^3\). Hence there is a
unital complete order isomorphism
\[
        \psi:\mathbb C^3\longrightarrow\mathcal S.
\]

Let \(e_1,e_2,e_3\) be the canonical minimal projections in
\(\mathbb C^3\), and let
\[
        B_i=\psi(e_i).
\]
Then
\[
        B_i\geq0,
        \qquad
        \sum_{i=1}^3B_i=I.
\]
The coordinate functionals of \(\psi^{-1}\) are states on
\(\mathcal S\). Extending them to states \(\omega_i\) on \(\mathsf M_3\), we
obtain
\[
        \omega_i(B_j)=\delta_{ij}.
\]
Let \(\rho_i\) be a density matrix representing \(\omega_i\). Since
\(0\leq B_i\leq I\),
\[
        \operatorname{Tr}(\rho_iB_i)=1
\]
implies that \(B_i\) acts as the identity on
\(\operatorname{supp}\rho_i\). Similarly,
\[
        \operatorname{Tr}(\rho_iB_j)=0,
        \qquad j\neq i,
\]
implies that \(B_j\) vanishes on \(\operatorname{supp}\rho_i\).

Choose a unit vector \(v_i\in\operatorname{supp}\rho_i\). Then
\[
        B_jv_i=\delta_{ij}v_i.
\]
The vectors \(v_1,v_2,v_3\) are mutually orthogonal. Indeed, for
\(i\neq j\),
\[
        \langle v_j,v_i\rangle
        =
        \langle v_j,B_iv_i\rangle
        =
        \langle B_iv_j,v_i\rangle
        =0.
\]
They therefore form an orthonormal basis of \(\mathbb C^3\), and
\[
        B_i=|v_i\rangle\langle v_i|.
\]
Consequently, the \(B_i\) commute and their linear span is a
commutative \(C^*\)-subalgebra of \(\mathsf M_3\).
\end{proof}

\begin{proposition}\label{prop:continuous-counterexample}
There exists a continuous, locally absolutely continuous,
Heisenberg CP-divisible family
\[
        \Lambda^*_t:\mathsf M_3\longrightarrow\mathsf M_3,
        \qquad
        \Lambda^*_0=\operatorname{id},
\]
and strict times \(t>s\) for which no    generalized inverse \({\Lambda_s^*}^-\)
of \(\Lambda^*_s\) makes $\Lambda^*_t{\Lambda_s^*}^-$ completely positive.
\end{proposition}

\begin{proof}
Let
\[
        E_i=|i\rangle\langle i|,
        \qquad i=1,2,3,
\]
and define the UCP map
\[
        \mathcal E(X)=\sum_{i=1}^3E_iXE_i .
\]
It is the conditional expectation onto the diagonal algebra {satisfying $\mathcal E^2=\mathcal E$}. For
\(0\leq q\leq1\), define
\[
        \Delta_q
        :=
        \mathcal E+q(\operatorname{id}-\mathcal E)
        =
        q\,\operatorname{id}+(1-q)\mathcal E.
\]
The maps \(\Delta_q\) are UCP and satisfy
\[
        \Delta_q\Delta_r=\Delta_{qr}.
\]
No, for \(0\leq t\leq1\), let
\[
        \Lambda^*_t=\Delta_{1-t}.
\]
Thus
\[
        \Lambda^*_0=\operatorname{id},
        \qquad
        \Lambda^*_1=\mathcal E.
\]

We now deform the diagonal range. Let \(J\) denote the \(3\times3\)
matrix whose entries are all equal to one, and define
\[
        U=I-\frac{2}{3}J.
\]
Since \(J^2=3J\), one has $U^*=U$ together with $U^2=I$. Consider the following semigroup
\[
        T_\tau
        :=
        \exp\!\left[\tau
        \bigl(\operatorname{Ad}_U-\operatorname{id}\bigr)\right],
        \qquad
        \operatorname{Ad}_U(X)=UXU^*.
\]
Because \(\operatorname{Ad}_U^2=\operatorname{id}\), this semigroup has
the explicit form
\[
        T_\tau
        =
        a_\tau\operatorname{id}
        +
        b_\tau\operatorname{Ad}_U,
\]
where
\[
        a_\tau=\frac{1+e^{-2\tau}}{2},
        \qquad
        b_\tau=\frac{1-e^{-2\tau}}{2}.
\]
In particular, every \(T_\tau\) is UCP and
\[
        T_\tau T_\sigma=T_{\tau+\sigma}.
\]

Define
\[
        \theta(t)=\min\{t-1,1\},
        \qquad t\geq1,
\]
and put
\[
        \Lambda^*_t=T_{\theta(t)}\mathcal E,
        \qquad t\geq1.
\]
Equivalently,
\[
        \Lambda^*_t=
        \begin{cases}
        \Delta_{1-t},
             & 0\leq t\leq1,\\[1mm]
        T_{t-1}\mathcal E,
             & 1\leq t\leq2,\\[1mm]
        T_1\mathcal E,
             & t\geq2.
        \end{cases}
\]
This family is continuous and locally absolutely continuous. Moreover, it is Heisenberg CP-divisible. Indeed, one may choose
\[
        W_{t,s}^*=
        \begin{cases}
        \displaystyle
        \Delta_{\frac{1-t}{1-s}},
             & 0\leq s\leq t<1,\\[3mm]
        T_{\theta(t)}\mathcal E,
             & 0\leq s<1\leq t,\\[1mm]
        T_{\theta(t)-\theta(s)},
             & 1\leq s\leq t.
        \end{cases}
\]
All these maps are UCP. Moreover,
\[
        W_{t,s}^*\Lambda^*_s=\Lambda^*_t.
\]
The second case $0\leq s<1\leq t$ follows from
\[
        \mathcal E\Delta_q=\mathcal E,
\]
whereas the third case $1\leq s\leq t$ follows from the semigroup law for \(T_\tau\).

For \(t\geq2\), the family is constant:
\[
        \Lambda^*_t=F:=T_1\mathcal E.
\]
Its range is
\[
        \mathcal S
        :=
        \operatorname{Im}F
        =
        \operatorname{span}\{A_1,A_2,A_3\},
\]
where
\[
        A_i=T_1(E_i)
        =aE_i+b\,UE_iU,
\]
with
\[
        a=\frac{1+e^{-2}}{2},
        \qquad
        b=\frac{1-e^{-2}}{2}.
\]
The map \(T_1\) is invertible: its eigenvalues  are \(1\) and \(e^{-2}\). Hence $\dim_{\mathbb C}\mathcal S=3$.

The operator system \(\mathcal S\) is not commutative{, since
\begin{equation}
    \bra1\bigl[A_1, A_2\bigr]\ket3 = \frac49{ab}\neq0.
\end{equation}
Therefore, by virtue of Lemma~\ref{lem:three-dimensional-range},} there is no UCP projection from \(\mathsf M_3\)
onto $\mathcal S=\operatorname{Im}F$.

Now choose any strict pair \(t>s\geq2\). Then
\[
        \Lambda^*_s=\Lambda^*_t=F.
\]
The identity map is a UCP propagator $F=\operatorname{id}  F$. 
Suppose, however, that there was a generalized inverse {\(F^-\)}
of \(F\) such  $W=FF^-$ is completely positive. Since
\[
        FF^-F=F,
\]
one would have
\[
        W^2
        =
        FF^-FF^-
        =
        (FF^-F)F^-
        =
        FF^-
        =
        W.
\]
Furthermore,
\[
        \operatorname{Im}W\subseteq\operatorname{Im}F=\mathcal S,
\]
while
\[
        W(F(X))=FF^-F(X)=F(X)
\]
shows that \(W\) fixes every element of \(\mathcal S\). Therefore $\operatorname{Im}W=\mathcal S$.
Since \(I\in\mathcal S\), one also has $W(I)=I$. Thus \(W\) would be a UCP projection onto \(\mathcal S\), which is impossible. Hence no generalized inverse of \(F\) produces a
completely positive propagator.
\end{proof}

\bibliographystyle{iopart-num}
\bibliography{biblio}

@article{CRS,
  title = {Divisibility and Information Flow Notions of Quantum Markovianity for Noninvertible Dynamical Maps},
  volume = {121},
  number = {8},
  pages = {080407},
  issn = {10797114},
  doi = {10.1103/PhysRevLett.121.080407},
  journal = {Physical Review Letters},
  author = {Chruściński, Dariusz and Rivas, Angel and St{\o}rmer, Erling},
  year = {2018}
}

@article{CC,
	title = {Construction of propagators for divisible dynamical maps},
	volume = {23},
	issn = {13672630},
	doi = {10.1088/1367-2630/abd43b},
	pages = {013009},
	journal = {New Journal of Physics},
	author = {Chakraborty, Ujan and Chruściński, Dariusz},
	year = {2021}
}

@book{GI_2,
	location = {New York},
	title = {Generalized Inverses},
	isbn = {0-387-00293-6},
	url = {http://link.springer.com/10.1007/b97366},
	doi = {10.1007/b97366},
	publisher = {Springer-Verlag},
	author = {Ben-Israel, Adi and Greville, Thomas N. E.},
	year = {2003},
	note = {Series Title: {CMS} Books in Mathematics},
}

@book{Paulsen,
author = {Paulsen, Vern},
doi = {10.1017/CBO9780511546631},
isbn = {9780521816694},
month = {feb},
publisher = {Cambridge University Press},
title = {{Completely Bounded Maps and Operator Algebras}},
url = {https://www.cambridge.org/core/product/identifier/9780511546631/type/book},
year = {2003}
}

@article{heinosaari_extending_2012,
	title = {Extending quantum operations},
	volume = {53},
	issn = {0022-2488, 1089-7658},
	url = {http://arxiv.org/abs/1205.0641},
	doi = {10.1063/1.4755845},
	number = {10},
	urldate = {2026-09-25},
	journal = {Journal of Mathematical Physics},
	author = {Heinosaari, Teiko and Jivulescu, Maria A. and Reeb, David and Wolf, Michael M.},
	month = oct,
	year = {2012},
	note = {arXiv:1205.0641 [math-ph]},
	pages = {102208}
}

@article{Filippov-ph-cov,
archivePrefix = {arXiv},
arxivId = {1911.09468},
author = {Filippov, S. N. and Glinov, A. N. and Lepp{\"{a}}j{\"{a}}rvi, L.},
doi = {10.1134/S1995080220040095},
eprint = {1911.09468},
issn = {1995-0802},
journal = {Lobachevskii Journal of Mathematics},
month = {apr},
number = {4},
pages = {617--630},
title = {{Phase Covariant Qubit Dynamics and Divisibility}},
volume = {41},
year = {2020}
}

@article{Teittinen2018,
author = {Teittinen, J and Lyyra, H and Sokolov, B and Maniscalco, S},
doi = {10.1088/1367-2630/aacc38},
issn = {1367-2630},
journal = {New Journal of Physics},
month = {jul},
number = {7},
pages = {073012},
title = {{Revealing memory effects in phase-covariant quantum master equations}},
volume = {20},
year = {2018}
}

@article{Smirne-ultimate,
	author = {Smirne, Andrea and Ko{\l}ody{\'{n}}ski, Jan and Huelga, Susana F. and Demkowicz-Dobrza{\'{n}}ski, Rafa{\l}},
	doi = {10.1103/PhysRevLett.116.120801},
	journal = {Phys. Rev. Lett.},
	number = {12},
	pages = {120801},
	title = {{Ultimate Precision Limits for Noisy Frequency Estimation}},
	volume = {116},
	year = {2016}
}

@book{Breuer-Petruccione,
	title = {The {Theory} of {Open} {Quantum} {Systems}},
	isbn = {0-19-921390-9},
	url = {https://academic.oup.com/book/27757},
	doi = {10.1093/acprof:oso/9780199213900.001.0001},
	publisher = {Oxford University PressOxford},
	author = {Breuer, Heinz-Peter and Petruccione, Francesco},
	month = jan,
	year = {2007},
}

@book{Rivas-Huelga-OQS,
	address = {Berlin, Heidelberg},
	title = {Open {Quantum} {Systems}},
	isbn = {978-3-642-23353-1},
	url = {https://link.springer.com/10.1007/978-3-642-23354-8},
	doi = {10.1007/978-3-642-23354-8},
	publisher = {Springer Berlin Heidelberg},
	author = {Rivas, Angel and Huelga, Susana F.},
	year = {2012},
	note = {Series Title: SpringerBriefs in Physics},
}

@book{Vacchini-OQS,
	address = {Cham},
	title = {Open {Quantum} {Systems}},
	isbn = {978-3-031-58217-2},
	url = {https://link.springer.com/10.1007/978-3-031-58218-9},
	doi = {10.1007/978-3-031-58218-9},
	publisher = {Springer Nature Switzerland},
	author = {Vacchini, Bassano},
	year = {2024},
	note = {Series Title: Graduate Texts in Physics},
}

@article{rivas_quantum_2014,
	title = {Quantum non-{Markovianity}: characterization, quantification and detection},
	volume = {77},
	issn = {0034-4885},
	url = {http://arxiv.org/abs/1405.0303},
	doi = {10.1088/0034-4885/77/9/094001},
	number = {9},
	journal = {Reports on Progress in Physics},
	author = {Rivas, Angel and Huelga, Susana F. and Plenio, Martin B.},
	month = sep,
	year = {2014},
	note = {arXiv: 1405.0303},
	pages = {094001},
}

@article{DeVega2017,
	title = {Dynamics of non-{Markovian} open quantum systems},
	volume = {89},
	issn = {15390756},
	doi = {10.1103/RevModPhys.89.015001},
	number = {1},
	journal = {Reviews of Modern Physics},
	author = {De Vega, Inés and Alonso, Daniel},
	year = {2017},
	note = {arXiv: 1511.06994},
	pages = {1--58},
}

@article{BLPV,
	title = {Colloquium: {Non}-{Markovian} dynamics in open quantum systems},
	volume = {88},
	issn = {15390756},
	doi = {10.1103/RevModPhys.88.021002},
	number = {2},
	journal = {Reviews of Modern Physics},
	publisher = {American Physical Society},
	author = {Breuer, Heinz Peter and Laine, Elsi Mari and Piilo, Jyrki and Vacchini, Bassano},
	month = apr,
	year = {2016},
}

@article{li_concepts_2018,
	title = {Concepts of quantum non-{Markovianity}: {A} hierarchy},
	volume = {759},
	issn = {03701573},
	doi = {10.1016/j.physrep.2018.07.001},
	journal = {Physics Reports},
	author = {Li, Li and Hall, Michael J.W. and Wiseman, Howard M.},
	year = {2018},
	note = {arXiv: 1712.08879},
	pages = {1--51},
}

@article{Chruscinski2022,
	title = {Dynamical maps beyond {Markovian} regime},
	volume = {992},
	issn = {03701573},
	url = {https://linkinghub.elsevier.com/retrieve/pii/S0370157322003428},
	doi = {10.1016/j.physrep.2022.09.003},
	journal = {Physics Reports},
	author = {Chruściński, Dariusz},
	month = dec,
	year = {2022},
	pages = {1--85},
}

@article{Settimo-SchroHeis,
	title = {Divisibility of {Dynamical} {Maps}: {Schrödinger} {Versus} {Heisenberg} {Picture}},
	volume = {7},
	issn = {2691-3399},
	url = {https://link.aps.org/doi/10.1103/6dt2-sq44},
	doi = {10.1103/6dt2-sq44},
	number = {1},
	journal = {PRX Quantum},
	author = {Settimo, Federico and Smirne, Andrea and Luoma, Kimmo and Vacchini, Bassano and Piilo, Jyrki and Chruściński, Dariusz},
	month = feb,
	year = {2026},
	note = {arXiv: 2506.08103},
	pages = {010340},
}

@article{andersson_finding_2007,
	title = {Finding the {Kraus} decomposition from a master equation and vice versa},
	volume = {54},
	issn = {0950-0340, 1362-3044},
	url = {http://arxiv.org/abs/0801.4100},
	doi = {10.1080/09500340701352581},
	number = {12},
	urldate = {2026-09-28},
	journal = {Journal of Modern Optics},
	author = {Andersson, Erika and Cresser, Jim D. and Hall, Michael J. W.},
	month = aug,
	year = {2007},
	note = {arXiv:0801.4100 [quant-ph]},
	pages = {1695--1716},
}

@article{cresser_master_2010,
	title = {Master equations with memory for systems driven by classical noise},
	volume = {283},
	copyright = {https://www.elsevier.com/tdm/userlicense/1.0/},
	issn = {00304018},
	url = {https://linkinghub.elsevier.com/retrieve/pii/S003040180901044X},
	doi = {10.1016/j.optcom.2009.10.052},
	number = {5},
	urldate = {2026-09-28},
	journal = {Optics Communications},
	author = {Cresser, J.D. and Facer, C.},
	month = mar,
	year = {2010},
	pages = {773--780},
}

@article{hou_singularity_2012,
	title = {Singularity of dynamical maps},
	volume = {86},
	issn = {1050-2947, 1094-1622},
	url = {http://arxiv.org/abs/1203.0053},
	doi = {10.1103/PhysRevA.86.012101},
	number = {1},
	urldate = {2026-09-28},
	journal = {Physical Review A},
	author = {Hou, S. C. and Yi, X. X. and Yu, S. X. and Oh, C. H.},
	month = jul,
	year = {2012},
	note = {arXiv:1203.0053 [quant-ph]},
	pages = {012101},
}

@article{fang_non-markovian_2018,
	title = {Non-{Markovian} dynamics of a qubit due to single-photon scattering in a waveguide},
	volume = {20},
	issn = {1367-2630},
	url = {https://iopscience.iop.org/article/10.1088/1367-2630/aaba5d},
	doi = {10.1088/1367-2630/aaba5d},
	number = {4},
	urldate = {2026-09-28},
	journal = {New Journal of Physics},
	author = {Fang, Yao-Lung L and Ciccarello, Francesco and Baranger, Harold U},
	month = apr,
	year = {2018},
	pages = {043035},
}

@article{hegde_open_2021,
	title = {Open quantum dynamics with singularities: {Master} equations and degree of non-{Markovianity}},
	volume = {104},
	issn = {2469-9926, 2469-9934},
	shorttitle = {Open quantum dynamics with singularities},
	url = {https://link.aps.org/doi/10.1103/PhysRevA.104.062403},
	doi = {10.1103/PhysRevA.104.062403},
	number = {6},
	urldate = {2026-09-28},
	journal = {Physical Review A},
	author = {Hegde, Abhaya S. and Athulya, K. P. and Pathak, Vijay and Piilo, Jyrki and Shaji, Anil},
	month = dec,
	year = {2021},
	pages = {062403},
}

@article{dabrowska_eternally_2021,
	title = {Eternally non-{Markovian} dynamics of a qubit interacting with a single-photon wavepacket},
	volume = {23},
	issn = {13672630},
	doi = {10.1088/1367-2630/ac3c60},
	number = {12},
	journal = {New Journal of Physics},
	author = {Dąbrowska, Anita and Chruściński, Dariusz and Chakraborty, Sagnik and Sarbicki, Gniewomir},
	year = {2021},
}

@article{siudzinska_non-markovianity_2022,
	title = {Non-{Markovianity} criteria for mixtures of noninvertible {Pauli} dynamical maps},
	volume = {55},
	issn = {1751-8113, 1751-8121},
	url = {https://iopscience.iop.org/article/10.1088/1751-8121/ac65c0},
	doi = {10.1088/1751-8121/ac65c0},
	number = {21},
	urldate = {2026-09-28},
	journal = {Journal of Physics A: Mathematical and Theoretical},
	author = {Siudzińska, Katarzyna},
	month = may,
	year = {2022},
	pages = {215201},
}

@article{jagadish_measure_2022,
	title = {Measure of invertible dynamical maps under convex combinations of noninvertible dynamical maps},
	volume = {106},
	issn = {2469-9926, 2469-9934},
	url = {https://link.aps.org/doi/10.1103/PhysRevA.106.012438},
	doi = {10.1103/PhysRevA.106.012438},
	number = {1},
	urldate = {2026-09-28},
	journal = {Physical Review A},
	author = {Jagadish, Vinayak and Srikanth, R. and Petruccione, Francesco},
	month = jul,
	year = {2022},
	pages = {012438},
}

@article{jagadish_noninvertibility_2023,
	title = {Noninvertibility and non-{Markovianity} of quantum dynamical maps},
	volume = {108},
	issn = {2469-9926, 2469-9934},
	url = {https://link.aps.org/doi/10.1103/PhysRevA.108.042202},
	doi = {10.1103/PhysRevA.108.042202},
	number = {4},
	urldate = {2026-09-28},
	journal = {Physical Review A},
	author = {Jagadish, Vinayak and Srikanth, R. and Petruccione, Francesco},
	month = oct,
	year = {2023},
	pages = {042202},
}

@misc{davidovic_nonperturbative_2026,
	title = {Nonperturbative {Resummation} of {Divergent} {Time}-{Local} {Generators}},
	url = {http://arxiv.org/abs/2603.25512},
	doi = {10.48550/arXiv.2603.25512},
	urldate = {2026-09-28},
	publisher = {arXiv},
	author = {Davidovic, Dragomir},
	month = jun,
	year = {2026},
	note = {arXiv:2603.25512 [quant-ph]},
}

@misc{cockett_generalized_2026,
	title = {Generalized {Inverses} of {Quantum} {Channels}: a categorical perspective},
	shorttitle = {Generalized {Inverses} of {Quantum} {Channels}},
	url = {http://arxiv.org/abs/2603.13946},
	doi = {10.48550/arXiv.2603.13946},
	urldate = {2026-09-28},
	publisher = {arXiv},
	author = {Cockett, Robin and Lemay, Jean-Simon Pacaud and Srinivasan, Priyaa Varshinee},
	month = mar,
	year = {2026},
	note = {arXiv:2603.13946 [quant-ph]},
}

@misc{blumenfeld_modifying_2025,
	title = {Modifying the {Time}-{Convolutionless} {Master} {Equation} via the {Moore}-{Penrose} {Pseudoinverse}},
	url = {http://arxiv.org/abs/2511.02556},
	doi = {10.48550/arXiv.2511.02556},
	urldate = {2026-09-28},
	publisher = {arXiv},
	author = {Blumenfeld, Caleb},
	month = nov,
	year = {2025},
	note = {arXiv:2511.02556 [quant-ph]},
}

@article{Arveson,
	title = {On subalgebras of \${C}{\textasciicircum}*\$-algebras},
	volume = {75},
	issn = {0002-9904},
	url = {https://www.ams.org/bull/1969-75-04/S0002-9904-1969-12293-7/},
	doi = {10.1090/S0002-9904-1969-12293-7},
	number = {4},
	urldate = {2026-09-28},
	journal = {Bulletin of the American Mathematical Society},
	author = {Arveson, William B.},
	month = jul,
	year = {1969},
	pages = {790--795},
}

@article{Tomiyama,
	title = {On the projection of norm one in \${W}{\textasciicircum} *\$-algebras},
	volume = {33},
	issn = {0386-2194},
	url = {https://projecteuclid.org/journals/proceedings-of-the-japan-academy-series-a-mathematical-sciences/volume-33/issue-10/On-the-projection-of-norm-one-in-W--algebras/10.3792/pja/1195524885.full},
	doi = {10.3792/pja/1195524885},
	number = {10},
	urldate = {2026-09-28},
	journal = {Proceedings of the Japan Academy, Series A, Mathematical Sciences},
	author = {Tomiyama, Jun},
	month = jan,
	year = {1957},
}

@article{ChoiEffros,
	title = {Injectivity and operator spaces},
	volume = {24},
	number = {2},
	journal = {Journal of functional analysis},
	author = {Choi, Man-Duen and Effros, Edward G.},
	year = {1977},
	pages = {156--209},
}

@misc{Stormer,
	title = {Extension of positive maps},
	url = {http://arxiv.org/abs/1805.05148},
	doi = {10.48550/arXiv.1805.05148},
	urldate = {2026-09-28},
	publisher = {arXiv},
	author = {Størmer, Erling},
	month = may,
	year = {2018},
	note = {arXiv:1805.05148 [math.OA]},
}

@article{chakraborty_information_2019,
	title = {Information flow versus divisibility for qubit evolution},
	volume = {99},
	issn = {2469-9926, 2469-9934},
	url = {https://link.aps.org/doi/10.1103/PhysRevA.99.042105},
	doi = {10.1103/PhysRevA.99.042105},
	number = {4},
	urldate = {2026-09-28},
	journal = {Physical Review A},
	author = {Chakraborty, Sagnik and Chruściński, Dariusz},
	month = apr,
	year = {2019},
	pages = {042105},
}

@article{braun_universal_2014,
	title = {A universal set of qubit quantum channels},
	volume = {47},
	copyright = {http://iopscience.iop.org/info/page/text-and-data-mining},
	issn = {1751-8113, 1751-8121},
	url = {https://iopscience.iop.org/article/10.1088/1751-8113/47/13/135302},
	doi = {10.1088/1751-8113/47/13/135302},
	number = {13},
	urldate = {2026-09-28},
	journal = {Journal of Physics A: Mathematical and Theoretical},
	author = {Braun, Daniel and Giraud, Olivier and Nechita, Ion and Pellegrini, Clément and Žnidarič, Marko},
	month = apr,
	year = {2014},
	pages = {135302},
}

@book{Stormer_book,
	address = {Berlin, Heidelberg},
	series = {Springer {Monographs} in {Mathematics}},
	title = {Positive {Linear} {Maps} of {Operator} {Algebras}},
	copyright = {https://www.springernature.com/gp/researchers/text-and-data-mining},
	isbn = {978-3-642-34368-1 978-3-642-34369-8},
	url = {https://link.springer.com/10.1007/978-3-642-34369-8},
	doi = {10.1007/978-3-642-34369-8},
	urldate = {2026-09-28},
	publisher = {Springer Berlin Heidelberg},
	author = {Størmer, Erling},
	year = {2013},
}

@book{Bhatia,
	title = {Positive {Definite} {Matrices}},
	isbn = {978-1-4008-2778-7},
	url = {https://www.degruyter.com/document/doi/10.1515/9781400827787/html},
	doi = {10.1515/9781400827787},
	publisher = {Princeton University Press},
	author = {Bhatia, Rajendra},
	month = dec,
	year = {2009},
}

@book{Horn,
	edition = {1},
	title = {Matrix {Analysis}},
	copyright = {https://www.cambridge.org/core/terms},
	isbn = {978-0-521-38632-6 978-0-521-30586-0 978-0-511-81081-7},
	url = {https://www.cambridge.org/core/product/identifier/9780511810817/type/book},
	doi = {10.1017/CBO9780511810817},
	urldate = {2026-09-28},
	publisher = {Cambridge University Press},
	author = {Horn, Roger A. and Johnson, Charles R.},
	month = dec,
	year = {1985},
}

@book{GI_1,
	address = {New York},
	series = {Wiley series in probability and mathematical statistics},
	title = {Generalized inverse of matrices and its applications},
	isbn = {978-0-471-70821-6},
	publisher = {Wiley},
	author = {Rao, C. Radhakrishna and Mitra, Sujit Kumar},
	year = {1971},
}

@book{GI_3,
	address = {New York, NY},
	title = {Projection {Matrices}, {Generalized} {Inverse} {Matrices}, and {Singular} {Value} {Decomposition}},
	isbn = {978-1-4419-9887-3},
	publisher = {Springer New York},
	author = {Yanai, Haruo and Takeuchi, Kei and Takane, Yoshio},
	year = {2011},
	note = {OCLC: 1120555938},
}

@article{chruscinski_measures_2011,
	title = {Measures of non-{Markovianity}: {Divisibility} versus backflow of information},
	volume = {83},
	copyright = {http://link.aps.org/licenses/aps-default-license},
	issn = {1050-2947, 1094-1622},
	shorttitle = {Measures of non-{Markovianity}},
	url = {https://link.aps.org/doi/10.1103/PhysRevA.83.052128},
	doi = {10.1103/PhysRevA.83.052128},
	number = {5},
	urldate = {2026-09-28},
	journal = {Physical Review A},
	author = {Chruściński, Dariusz and Kossakowski, Andrzej and Rivas, Angel},
	month = may,
	year = {2011},
	pages = {052128},
}

@article{Chiribella,
	title = {Counterexamples to the extendibility of positive unital norm-one maps},
	volume = {663},
	issn = {00243795},
	url = {https://linkinghub.elsevier.com/retrieve/pii/S0024379523000034},
	doi = {10.1016/j.laa.2023.01.003},
	urldate = {2026-09-28},
	journal = {Linear Algebra and its Applications},
	author = {Chiribella, Giulio and Davidson, Kenneth R. and Paulsen, Vern I. and Rahaman, Mizanur},
	month = apr,
	year = {2023},
	pages = {102--115},
}

@article{Amato,
	title = {Asymptotic {Dynamics} in the {Heisenberg} {Picture}: {Attractor} {Subspace} and {Choi}-{Effros} {Product}},
	volume = {32},
	issn = {1230-1612, 1793-7191},
	shorttitle = {Asymptotic {Dynamics} in the {Heisenberg} {Picture}},
	url = {https://www.worldscientific.com/doi/10.1142/S1230161225500143},
	doi = {10.1142/S1230161225500143},
	number = {03},
	urldate = {2026-09-28},
	journal = {Open Systems \& Information Dynamics},
	author = {Amato, Daniele and Facchi, Paolo and Konderak, Arturo},
	month = sep,
	year = {2025},
	pages = {2550014},
}

@article{Choi1974,
	title = {A schwarz inequality for positive linear maps on \${C}{\textasciicircum}\{{\textbackslash}ast\}\$-algebras},
	volume = {18},
	issn = {0019-2082},
	url = {https://projecteuclid.org/journals/illinois-journal-of-mathematics/volume-18/issue-4/A-schwarz-inequality-for-positive-linear-maps-on-Cast-algebras/10.1215/ijm/1256051007.full},
	doi = {10.1215/ijm/1256051007},
	number = {4},
	urldate = {2026-09-28},
	journal = {Illinois Journal of Mathematics},
	author = {Choi, Man-Duen},
	month = dec,
	year = {1974},
}

@article{MD1,
	title = {The multiplicative domain in quantum error correction},
	volume = {42},
	issn = {1751-8113, 1751-8121},
	url = {https://iopscience.iop.org/article/10.1088/1751-8113/42/24/245303},
	doi = {10.1088/1751-8113/42/24/245303},
	number = {24},
	urldate = {2026-09-28},
	journal = {Journal of Physics A: Mathematical and Theoretical},
	author = {Choi, Man-Duen and Johnston, Nathaniel and Kribs, David W},
	month = jun,
	year = {2009},
	pages = {245303},
}

@article{MD2,
	title = {Generalized {Multiplicative} {Domains} and {Quantum} {Error} {Correction}},
	volume = {139},
	issn = {0002-9939},
	url = {https://www.jstor.org/stable/41059318},
	number = {2},
	urldate = {2026-09-28},
	journal = {Proceedings of the American Mathematical Society},
	publisher = {American Mathematical Society},
	author = {Johnston, Nathaniel and Kribs, David W.},
	year = {2011},
	pages = {627--639},
}

@article{GKS,
  author  = {Gorini, V. and Kossakowski, A. and Sudarshan, E. C. G.},
  title   = {Completely positive dynamical semigroups of {N}-level systems},
  journal = {J. Math. Phys.},
  volume  = {17},
  pages   = {821--825},
  year    = {1976}
}

@article{L,
  author  = {Lindblad, G.},
  title   = {On the generators of quantum dynamical semigroups},
  journal = {Commun. Math. Phys.},
  volume  = {48},
  pages   = {119--130},
  year    = {1976}
}

@misc{Federico_Tasks,
	title = {Schrödinger and {Heisenberg} non-{Markovianity} in quantum information tasks},
	url = {http://arxiv.org/abs/2606.30751},
	author = {Settimo, Federico and Luoma, Kimmo and Piilo, Jyrki and Vacchini, Bassano and Chruściński, Dariusz and Smirne, Andrea},
	month = jun,
	year = {2026},
	note = {arXiv: 2606.30751},
}

@article{Pollock2018,
  author  = {F. A. Pollock and C. Rodr{\'i}guez-Rosario and
             T. Frauenheim and M. Paternostro and K. Modi},
  title   = {Operational Markov Condition for Quantum Processes},
  journal = {Physical Review Letters},
  volume  = {120},
  pages   = {040405},
  year    = {2018},
  doi     = {10.1103/PhysRevLett.120.040405}
}

\end{document}